%% file: main.tex
\documentclass{article}
\usepackage{amsmath}
\usepackage{wright}
\usepackage{tikz}
\usetikzlibrary{matrix}

\newtheorem*{theorem*}{Theorem}

\newcommand{\dU}{\mathrm{d}U}
\newcommand{\reg}[1]{\mathsf{#1}}

\newcommand{\Sym}[2]{\mathrm{Sym}^{(#1,#2)}}
\newcommand{\Pisym}[2]{\Pi_{\mathrm{Sym}}^{(#1,#2)}}

\newcommand{\Werner}{\mathcal{C}_{\mathrm{W}}}
\newcommand{\Purify}{\mathrm{RPC}}
\newcommand{\states}{\mathcal{D}}
\newcommand{\Schur}{\mathcal{U}_{\mathrm{Sch}}}

\newcommand{\PCT}{\mathcal{C}_{\mathrm{PCT}}}

\newcommand{\Fid}{\mathrm{F}}

\newcommand{\dP}{\mathrm{d}P}

\newcommand{\Par}{\mathrm{Par}}

\newcommand{\SSYT}{\mathrm{SSYT}}

\newcommand{\Lift}{\mathrm{Lift}}
\newcommand{\WSS}{\mathrm{WSS}}

\newcommand{\sh}{\mathrm{shape}}
\newcommand{\Proj}{\mathrm{Proj}}

\newcommand{\Sep}{\mathrm{Sep}}

\newsavebox{\originalqed}
\sbox{\originalqed}{\qedsymbol}
\let\originalBox\Box
\makeatletter
\renewcommand{\Box}{\mathpalette\loweredBox@{}}
\newcommand{\loweredBox@}[2]{%
    \raisebox{-0.12\height}{$\m@th#1\originalBox$}%
}
\makeatother
\renewcommand{\qedsymbol}{\usebox{\originalqed}}

\title{Optimal cloning of mixed states}
\author{Marco Fanizza\thanks{Inria, T\'el\'ecom Paris -- LTCI, Institut Polytechnique de Paris, Palaiseau, France. \texttt{marco.fanizza@inria.fr}}  \and Dmitry Grinko\thanks{QuSoft \& Institute for Logic, Language and Computation \& Korteweg-de Vries Institute for Mathematics, University of Amsterdam, The Netherlands. \texttt{d.grinko@uva.nl}} \and Thilo Scharnhorst\thanks{UC Berkeley. \texttt{\{thilo,jspilecki\}@berkeley.edu}} \and Jack Spilecki\footnotemark[3]}
\date{}

\begin{document}

\maketitle

\begin{abstract}
We consider the problem of \emph{approximate cloning} of quantum states: given $n$ copies of an unknown state $\rho \in \C^{d \times d}$, prepare an $(n+k)$-copy state with high fidelity to $\rho^{\otimes (n+k)}$. Werner's pure state cloner is the optimal channel for the pure state case, and shows that $n = \Theta(kd/\eps)$ copies are necessary and sufficient to clone $k$ additional copies of an unknown pure state to fidelity $1-\eps$. The random purification channel gives a straightforward extension of Werner's cloner to mixed state inputs: given $n$ copies of a mixed state, randomly purify your input, apply Werner's channel in the larger Hilbert space, and then trace out the auxiliary registers. This gives a mixed state cloner using $n = O(krd/\eps)$ copies to clone rank-$r$ states. Can one do any better? We show that the answer is no: one must use $n = \Omega(krd/\eps)$ copies. We prove our lower bound by studying the special case of \emph{projector cloning}, in which the input state $\rho$ is promised to be of the form $P/r$, where $P$ is a rank-$r$ orthogonal projector.

As a further application of our techniques, we consider the closely related problem of \emph{approximate transposition} of quantum states, where one seeks to convert $\rho^{\otimes n}$ to a $k$-copy state with high fidelity to~$(\rho^T)^{\otimes k}$. Here, we again show $n = \Theta(krd/\eps)$ copies are necessary and sufficient for this task.
\end{abstract}

\newpage

\hypersetup{linktocpage}
\tableofcontents
\thispagestyle{empty}

\newpage

\section{Introduction}

\input{intro.tex}

\section{Preliminaries}

\input{prelims.tex}

\section{Sample-optimal cloning of mixed states} \label{sec:cloning}

\input{cloning.tex}

\section{Sample-optimal transposition of mixed states} \label{sec:transposition}

\input{transposition.tex}

\section*{Acknowledgments}
D.G.\ acknowledges support by NWO grant NGF.1623.23.025 (“Qudits in theory and experiment”).
J.S.\ is supported by John Wright's NSF CAREER award CCF-233971, and would like to thank Omar Alrabiah, Angelos Pelecanos, and particularly John Wright for many valuable discussions.

\paragraph{AI statement.} We used AI tools throughout the course of this project to discuss ideas, perform computations, prove low-level results, and assist with the preparation of the manuscript. The
central ideas of this work were human-generated, with the exception of the recursion used to prove \Cref{prop:untransposed_bound}.

\bibliographystyle{alpha}
\bibliography{wright}

\end{document}

%% file: intro.tex
How many copies of a quantum state does it take to make one more? If we require that the ``clone'' be exact, then the no-cloning theorem \cite{WZ82,Die82} famously tells us that no finite number of copies suffices: there is no channel $\calC$ such that $\calC(\rho^{\otimes n}) = \rho^{\otimes (n+k)}$ for all states $\rho$, for any fixed $k \geq 1$, no matter how large $n$ is taken to be. The no-cloning theorem is a foundational result that shapes the field of quantum information theory. 

However, exactness is a strict standard, and not one that quantum computations meet in practice anyway. The no-cloning theorem says nothing about the possibility of \emph{approximate cloning}, and it is the approximate version of our question that determines how no-cloning \emph{really} limits us. Therefore, we have a natural question: for fixed $k \geq 1$, how many copies $n$ of an unknown quantum state $\rho \in \C^{d \times d}$ do we need to prepare a state that approximates $\rho^{\otimes (n+k)}$ to some desired accuracy $\epsilon$?

In this work, we measure the quality of the output state by its fidelity\footnote{We will use the ``squared-fidelity'' convention: $\Fid(\rho, \sigma) \coloneq \norm{ \sqrt{\rho} \sqrt{\sigma} }_1^2$, which reduces to $| \braket{u}{v} |^2$ for pure states.} with $\rho^{\otimes (n+k)}$, and seek channels that attain the optimal scaling of $n$. That is, we want to construct channels $\calC$, which we call \emph{cloning channels} or \emph{cloners}, which use the fewest number of copies $n$ while maintaining the guarantee
\begin{equation*}
    \min_{\rho} \Big[ \Fid\big( \calC(\rho^{\otimes n}), \rho^{\otimes (n+k)} \big)\Big] \geq 1 - \eps.
\end{equation*}

The pure state special case of approximate cloning has been completely solved. Following the first constructions of approximate cloners \cite{BH96,GM97}, Werner constructed the unique optimal pure state cloner for arbitrary dimension $d$ \cite{Wer98}. On input $\ketbra{u}^{\otimes n}$, Werner's cloner outputs the state
\begin{equation*}
    \Werner\big( \ketbra{u}^{\otimes n} \big) = \frac{d[n]}{d[n+k]} \cdot \Pisym{d}{n+k} \cdot \big( \ketbra{u}^{\otimes n} \otimes I^{\otimes k} \big) \cdot \Pisym{d}{n+k}. 
\end{equation*}
Here, $\Pisym{d}{m}$ is the projector onto the symmetric subspace on $(\C^d)^{\otimes m}$, and $d[m]$ is the dimension of the symmetric subspace.
For every pure state input, this cloner achieves fidelity $d[n]/d[n+k]$, which Werner proved is optimal, and uniquely achieved by his construction.
By analyzing this ratio, it can be shown that $n = O(kd/\eps)$ copies are sufficient, and $n = \Omega(kd/\eps)$ copies are required, in order to successfully clone $k$ extra copies of an unknown pure state.
Thus, the sample complexity of pure state cloning is settled as $n = \Theta(kd/\eps)$.
Since then, despite substantial progress on related problems --- such as mixed state broadcasting \cite{DMP05,BDMP06,Chiribella_2006,DF07}, asymmetric cloning \cite{Cer00,FFC05,IAG+05,NPR21,NPR23}, and cloning special families of states \cite{BCI+01,GM06,Chiribella_2014,Chiribella_2013,BCM26} --- there is still no comparable picture for approximate cloning of mixed states.
In the words of Scarani et al. \cite{SIGA05}, \emph{``The optimal cloning of mixed states is thus a completely open domain''}, and this has remained true since.

However, the recently introduced \emph{random purification channel} suggests a natural way to close this gap. This channel, introduced in \cite{TWZ25}, takes $n$ copies of any unknown mixed state $\rho$ to $n$ copies of the same, albeit random, purification of $\rho$ on a larger Hilbert space. The random purification channel can be used to give straightforward reductions from many mixed state tasks to their often far simpler pure state special cases. Simply: use it to obtain copies of a purification, carry out the pure state protocol in the larger Hilbert space, and then trace out the auxiliary purifying registers.

Applying this framework to Werner's pure state cloner gives a natural mixed state cloner which first appeared in the literature in \cite{LTHC26}. The cloner schematically acts as
\begin{equation*}
    \rho_{\reg{A}}^{\otimes n} \xlongrightarrow{\mathrm{Purify}} \ketbra{\rho}_{\reg{AB}}^{\otimes n} \xlongrightarrow{\mathrm{Clone}}  \ketbra{\rho}_{\reg{AB}}^{\otimes (n+k)} \xlongrightarrow {\mathrm{Trace}}  \rho^{\otimes (n+k)}_{\reg{A}}.
\end{equation*}
We call this the \emph{PCT cloner}, for the steps of its action. Of course, the pure state cloning in the second step is approximate, so the final state is approximate as well, but with enough copies $n$ to perform the cloning well in the enlarged Hilbert space, the final mixed state is close to the desired output as well, by data processing. In particular, for rank-$r$ inputs, the purification $\ket{\rho}$ can be taken to live in an $rd$-dimensional space, and then Werner's cloner requires only $n = O(krd/\eps)$ copies as input. This begs the question:
\begin{center}
\textit{Is the PCT cloner sample-optimal for mixed state cloning?}
\end{center}
That is, must one use $\Omega(krd/\eps)$ copies?

There is a reasonable case for expecting either resolution. On the one hand, against optimality: by cloning an entire \emph{purification} of $\rho$, the PCT cloner appears to do much more than necessary, only to throw away the purification registers. One might also look to the classical special case, where the natural analogue of cloning --- \emph{sample amplification} --- is an easier task than learning a classical probability distribution \cite{AGSV24,AGH+24,HJW15}. On the other hand, in favor of optimality: for tomography, the RPC-based algorithms might also seem to do more than necessary by learning a purification, but nevertheless yield sample-optimal algorithms \cite{PSTW25,Yue23,SSW25}. Moreover, optimal pure state cloning is structurally tied to optimal pure state tomography \cite{Hay98}, via Chiribella's equation \cite{Chi11} (see also~\cite{Bruss_1998,Bae06,Chiribella_2006}), and for pure states, cloning one additional copy has the same sample complexity as learning the state, $\Theta(d/\eps)$\footnote{We stress that this is an equality of sample complexities, not an operational reduction: tomography consumes all the copies and leaves only a classical estimate, whereas the optimal cloner is acting coherently and is not a measure-and-prepare channel.}. If the correspondence between the complexities of cloning and tomography is expected to survive into the mixed state setting, then it seems reasonable to expect that cloning one additional copy of a mixed state requires $\Theta(rd/\eps)$ copies.

In this work, we resolve our question, and determine the sample complexity of mixed state cloning.  

\begin{theorem}[The PCT cloner is sample-optimal] \label{thm:main_result_intro}
    Suppose $\calC: \states(\C^d)^{\otimes n} \to \states(\C^d)^{\otimes (n+k)}$ is a cloning channel such that
    \begin{equation*}
         \Fid \big(  \calC(\rho^{\otimes n}) , \rho^{\otimes (n+k)}\big) \geq 1-\eps,
    \end{equation*}
    for all states $\rho$ of rank at most $r$. Then we must have $n = \Omega(krd/\eps)$, for $d \geq 2$, and $\epsilon \leq \epsilon_0$, where $\epsilon_0$ is a universal constant. Moreover, the PCT cloner achieves this guarantee with $n = O(krd/\eps)$ copies. Thus, the sample complexity of cloning $k$ additional copies of any unknown rank-$r$ input state $\rho^{\otimes n}$ to fidelity $1-\eps$ is $\Theta(krd/\eps)$. 
\end{theorem}

Our lower bound holds for an especially simple class of states: \emph{projector states}. These are the states of the form $\rho = P/r$, for $P$ a rank-$r$ orthogonal projector. Previous work has shown that these states form hard examples for tomographic tasks \cite{HHJ+16,SSW25, CLW26}; we show that they are also hard for cloning. This sheds light on what makes cloning quantum states difficult: even for states with completely known spectra, correctly reproducing the support already necessitates the full sample complexity.

\subsection{Technical overview}

\paragraph{An upper bound for cloning.} Our upper bound is by a reduction from mixed state cloning to pure state cloning, via the recently introduced random purification channel \cite{TWZ25}. 


Suppose that $\rho \in \C^{d \times d}$ is a quantum state with rank at most $r$. The random purification channel $\Purify^{(d,r,n)}$ takes in $n$ copies of $\rho$, and outputs the state
\begin{equation*}
\Purify^{(d,r,n)}(\rho_{\reg{A}}^{\otimes n}) = \E_{\ket{\brho}} \big[ \ketbra{\brho}_{\reg{AB}}^{\otimes n} \big].
\end{equation*}
Here, we are labeling the $n$ initial registers collectively $\reg{A}$, and the $n$ auxiliary $r$-dimensional registers we must add to support the purifications collectively $\reg{B}$. Also, by \emph{random purification}, we mean a purification distributed as $\ket{\brho} = (I \otimes \bU) \cdot \ket{\rho_0}$, for a Haar random $\bU \in U(r)$, and any fixed purification $\ket{\rho_0}$. It is not so important for us that we produce \emph{this particular} distribution of random purifications, more so that we are able to prepare \emph{some} distribution over purifications. 

The output of the RPC behaves as a classical mixture over purifications of $\ket{\brho}^{\otimes n}$. So from here, the overall idea of the upper bound is simple: use the Werner cloner to clone the purification, obtaining a state close to $\ket{\brho}_{\reg{AB}}^{\otimes (n+k)}$, and then trace out the auxiliary registers, to obtain a state close to $\rho_{\reg{A}}^{\otimes (n+k)}$. 

{
\floatstyle{boxed} 
\restylefloat{figure}
\begin{figure}[H]
Given $n$ copies of a rank-$r$ mixed state $\rho_{\reg{A}}$:
\begin{enumerate}
    \item Apply $\Purify^{(d,r,n)}$ to prepare $n$ copies of a random purification $\ket{\brho}_{\reg{AB}} \in \C^d \otimes \C^r$. Set $D = d \cdot r$. 
    \item Apply $\Werner^{(D, n, k)}$, yielding a mixed state $\bsigma_{\reg{AB}} \in \states( \C^d \otimes \C^r)^{\otimes (n+k)}$. 
    \item Trace out the auxiliary registers, and output $\Tr_{\reg{B}}(\bsigma_{\reg{AB}}) \in \states(\C^d)^{\otimes (n+k)}$. 
\end{enumerate}
\caption{A mixed state cloning channel, which we call the \emph{PCT cloner}, due to its action: \emph{purify--clone--trace}. We denote the PCT cloner by $\PCT^{(d,n,k)}$. As usual, we will sometimes drop the superscript's parameters when these are clear from context.}
\label{fig:reduction_tech_overview}
\end{figure}
}

Thus, mixed state cloning follows easily from the combination of the random purification channel and pure state cloning. This natural strategy has already been described in the literature \cite{LTHC26}. We show the following result.

\begin{proposition}[Sample complexity of the PCT cloner] \label{prop:upper_bound_PCT_tech_overview}
    With $n = O(krd/\eps)$ copies of a mixed state $\rho \in \states(\C^d)$, the PCT cloner $\PCT^{(d,n,k)}$ produces an $(n+k)$-copy state with fidelity at least $1-\eps$ with $\rho^{\otimes (n+k)}$. 
\end{proposition}

The argument is simple, and we sketch the complete proof here. First suppose, after applying the RPC, we have copies of some purification $\ket{\brho}$. Werner's cloner then produces a state with fidelity
\begin{equation*}
\Fid \Big( \Werner\big(\ketbra{\brho}^{\otimes n} \big) , \ketbra{\brho}^{\otimes (n+k)}\Big) = \frac{D[n]}{D[n+k]}.
\end{equation*}
Here, $D[n]$ is the dimension of the symmetric subspace in dimension $D = d\cdot r$. It can be shown that $D[n]/D[n+k] \geq 1 - kD/n$, so in particular, the fidelity is at least $1-\epsilon$ for $n = O(kD/\epsilon) = O(krd/\eps)$. But then, since fidelity only increases upon tracing out the auxiliary registers,
\begin{equation*}
     \Fid \Big( \tr_{\reg{B}} \big(\Werner\big(\ketbra{\brho}^{\otimes n} \big)\big), \rho^{\otimes (n+k)}\Big) = \Fid \Big( \tr_{\reg{B}} \big(\Werner\big(\ketbra{\brho}^{\otimes n} \big)\big), \tr_{\reg{B}}\big(\ketbra{\brho}^{\otimes (n+k)}\big)\Big) \geq 1-\epsilon. 
\end{equation*}
This equation holds for any such purification $\ket{\brho}$, so in particular, it can also be shown to hold for the mixture over purifications, using the concavity of the (square-root) fidelity. Therefore, the following algorithm successfully clones an unknown rank-$r$ mixed state using $n = O(krd/\eps)$ copies. 

\paragraph{A lower bound for cloning.} Our lower bound is inspired by a lower bound for learning in fidelity given in \cite{SSW25}. In that work, the authors considered \emph{projector states}, states of the form $\rho = P/r$, for $P$ an orthogonal rank-$r$ projector. Projector states are a convenient family of states to analyze, owing to their simplicity and structure. Nevertheless, the authors showed that learning to fidelity $1-\eps$ requires $n = \Omega(rd/\eps)$ copies (given some simple restrictions on $r, d, \eps$). This matches the $O(rd/\eps)$ upper bound for learning rank-$r$ states \cite{PSW25}. Thus, the projector states were shown to form a set of hard instances for learning in fidelity. 

In this work, we likewise show that projector states form a family of hard instances for cloning.

\begin{proposition}[A lower bound for cloning projector states]\label{prop:projector_lower_bound_tech_overview}
Any channel $\calC: \states(\C^d)^{\otimes n} \to \states(\C^d)^{\otimes (n+k)}$ which clones rank-$r$ projector states to fidelity $1-\eps$ requires at least $n = \Omega(krd/\epsilon)$ copies as input, for $d \geq 2$, $r \leq d/2$, and $\epsilon \leq 1/16$. 
\end{proposition}

Our main result, \Cref{thm:main_result_intro}, follows from our upper bound and this lower bound, applied using projector states of rank $\min(r, d/2)$. 


We now give an overview of how we show our lower bound. Our first step is to reduce from fidelity, which is non-linear and difficult to analyze for mixed state generally, to a friendlier linear proxy.
We define a simple quantity we call the \emph{overlap}: for a projector $\Pi$ of rank $R$, the overlap with a quantum state $\sigma$ is just $\tr(\sigma \Pi)$. By an application of Cauchy-Schwarz, it can be shown that 
\begin{equation*}
    \tr( \sigma \cdot \Pi) \geq \Fid( \sigma, \Pi/R).
\end{equation*}
We apply this with $\Pi \leftarrow P^{\otimes (n+k)}$, and $\sigma \leftarrow \calC(\rho^{\otimes n})$. But this shows that if $\calC$ can clone projector states to fidelity $1-\eps$, then we also necessarily have
\begin{equation} \label{eq:tech_overview_1}
    \tr\Big( \calC(\rho^{\otimes n}) \cdot P^{\otimes (n+k)} \Big) \geq 1-\eps, 
\end{equation}
for any projector state $\rho = P/r$. So it suffices to show that cloning projectors in overlap leads to a lower bound on $n$. Note that cloning in overlap is seemingly a much easier task than cloning in fidelity: for overlap, all that is asked of a cloner is that it puts most of its mass on the right support only; fidelity asks that the cloner also has a balanced spectrum on that support. Nevertheless, cloning in overlap still requires $n=\Omega(krd/\eps)$, matching the lower bound for cloning in fidelity.

We upper bound the worst-case overlap by the average-case overlap
\begin{equation*}
    \min_{\rho} \Big[ \tr\big( \calC(\rho^{\otimes n}) \cdot P^{\otimes (n+k)} \big)\Big] \leq \int_P \tr\big( \calC(\rho^{\otimes n}) \cdot P^{\otimes (n+k)} \big) \cdot \dP.
\end{equation*}
Here, $\dP$ is the measure on projectors induced by the Haar measure. Furthermore, we can rewrite the right-hand side in terms of the Choi--Jamiołkowski state $J({\calC})$ corresponding to $\calC$:
\begin{equation*}
    \tr \Big( J({\calC}) \cdot \int_P \rho^{T,\otimes n} \otimes P^{\otimes (n+k)} \cdot \dP \Big) = \frac{1}{r^n} \tr \Big( J({\calC}) \cdot \int_P P^{T,\otimes n} \otimes P^{\otimes (n+k)} \cdot \dP \Big).
\end{equation*}
The integral over projectors is a highly symmetric object, and can be approached using representation theory. In particular, if we rewrite it in the Schur basis, we get
 \begin{equation*}
    \int P^{T,\otimes n} \otimes P^{\otimes (n+k)} \cdot \dP \cong \bigoplus_{\substack{\lambda \in \Par(n,d) \\ \mu \in \Par(n+k,d)}} I_{\calP_\lambda} \otimes I_{\calP_{\mu}} \otimes \int_P q^d_\lambda(P)^T \otimes q^d_\mu(P) \cdot \dP
\end{equation*}
We write $(\calP_\lambda, p_\lambda)$ for the irrep of $S_n$ corresponding to a partition $\lambda$, and $(\calQ^d_\lambda, q^d_\lambda)$ for the irrep of $U(d)$ corresponding to $\lambda$. Here, we will also specify that we will always be taking our transposes in the Schur basis. Trace-preserving constraints can then be used to derive the following upper bound.
\begin{equation*}
    \frac{1}{r^n} \tr \Big( J({\calC}) \cdot \int_P P^{T,\otimes n} \otimes P^{\otimes (n+k)} \cdot \dP \Big) \leq \sum_{\substack{\lambda \vdash n \\ \ell(\lambda) \leq r}} \frac{\dim(\calP_\lambda) \cdot \dim(\calQ^d_\lambda)}{r^n} \cdot \max_{ \substack{\mu \vdash n+k \\ \ell(\mu) \leq r}} \norm{ \int_P q_\lambda(P)^T \otimes q_\mu(P) \cdot \dP }_\infty.
\end{equation*}
We can reinterpret this in terms of weak Schur sampling. The quantity $\dim(\calP_\lambda) \cdot \dim(\calQ_\lambda^r)/r^n$ is equal to the probability of obtaining the diagram $\lambda$ upon weak Schur sampling $\rho^{\otimes n}$. Consequently, the upper bound above can be rewritten as
\begin{equation} \label{eq:tech_overview_2}
\E_{\blambda} \bigg[ \frac{s_{\blambda}(1^d)}{s_{\blambda}(1^r)} \cdot \max_{ \substack{\mu \vdash n+k \\ \ell(\mu) \leq r}} \norm{ \int_P q_{\blambda}(P)^T \otimes q_\mu(P) \cdot \dP }_\infty \bigg].
\end{equation}
The quantity in the brackets tells us how large the overlap could possibly be, having obtained $\blambda$. Intuitively, the maximum then reflects that the optimal Choi state, conditioned on having measured $\blambda$, will align with the largest eigenvector of 
\begin{equation*}\bigoplus_{ \mu \in \Par(n+k,d)} I_{\calP_{\blambda}} \otimes I_{\calP_{\mu}} \otimes \int_P q^d_{\blambda}(P)^T \otimes q^d_\mu(P) \cdot \dP.
\end{equation*}

We proceed by proving the following lemma, which is our main technical result. 

\begin{lemma}[The partially-transposed integral bound] \label{lem:transposed_integral_tech_overview}
    Let $\lambda \in \Par(n,d)$ and $\mu \in \Par(n+k,d)$, with $\ell(\lambda) \leq r$ and $\ell(\mu) \leq r$. Then we have
    \begin{equation*}
        \int q_{\lambda}^d(P)^T \otimes q_{\mu}^d(P) \cdot \dP \preceq \frac{s_{\lambda \cup \mu}(1^r)}{s_{\lambda \cup \mu}(1^d)} \cdot I_{\calQ^d_\lambda} \otimes I_{\calQ^d_\mu}. 
    \end{equation*}
\end{lemma}

Here, $\lambda \cup \mu$ is the \emph{union} of the Young diagrams $\lambda$ and $\mu$, viewed as sets of boxes. Alternatively, it is the diagram with $(\lambda \cup \mu)_i = \max( \lambda_i, \mu_i)$. Also, $s_\lambda$ is the Schur polynomial corresponding to $\lambda$. We will circle back to intuition for this bound, and how we prove it in the next part of the technical overview. For now, we will briefly describe how we use this to complete our lower bound proof. From the partially-transposed integral bound, we get
\begin{equation*}
    \eqref{eq:tech_overview_2} \leq \E_{\blambda} \bigg[ \frac{s_{\blambda}(1^d)}{s_{\blambda}(1^r)} \cdot \max_{ \substack{\mu \vdash n+k \\ \ell(\mu) \leq r}} \frac{s_{\blambda \cup \mu}(1^r)}{s_{\blambda \cup \mu}(1^d)} \bigg].
\end{equation*}
We then show that the $\mu$ which achieves the maximal value is $\mu = \blambda + k\cdot e_1$, i.e.\ the diagram obtained by starting with $\blambda$, and adding $k$ boxes to its first row. We thus obtain the inequality
\begin{equation} \label{eq:tech_overview_3}
    1 - \epsilon \leq \E_{\blambda} \bigg[ \frac{s_{\blambda}(1^d)}{s_{\blambda}(1^r)} \cdot \frac{s_{\blambda + k\cdot e_1}(1^r)}{s_{\blambda + k \cdot e_1}(1^d)} \bigg].
\end{equation}
Analyzing this inequality allows us to conclude that $n \geq krd/8\eps$, at least for certain parameters: $d \geq 2$, $r \leq d/2$ and $\eps \leq 1/16$. We note that $d \geq 2$ is necessary, as are $r < d$ and $\eps < 1$, to avoid cases where cloning is trivial, but we have not attempted to optimize our parameter restrictions, as these suffice for our purposes. 

\paragraph{The partially-transposed integral bound.} We begin by giving some intuition for the partially-transposed integral bound. First consider the case when $\lambda \vdash n$ is a subset of $\mu \vdash n+k$. In this case, the bound is easy to show. This is because $q^d_\lambda(P)$ is actually itself a projector on $\calQ^d_\lambda$, and therefore upper bounded in the PSD order by $I_{\calQ^d_\lambda}$. So
\begin{equation} \label{eq:tech_overview_4}
    \int q_{\lambda}^d(P)^T \otimes q_{\mu}^d(P) \cdot \dP \preceq \int I_{\calQ^d_\lambda} \otimes q_{\mu}^d(P) \cdot \dP = \frac{s_\mu(1^r)}{s_\mu(1^d)} \cdot I_{\calQ^d_\lambda} \otimes I_{\calQ^d_\mu} = \frac{s_{\lambda \cup \mu}(1^r)}{s_{\lambda \cup \mu}(1^d)} \cdot I_{\calQ^d_\lambda} \otimes I_{\calQ^d_\mu},
\end{equation}
using Schur's lemma in the second-last step, and that $\mu = \lambda \cup \mu$, since $\lambda \subseteq \mu$, in the last. So the partially-transposed integral bound holds in the special case where $\lambda \subseteq \mu$. 

It is, however, natural to expect that the output diagram $\mu$ really does contain the input diagram $\lambda$, for any good cloner. Indeed, since in cloning our input is of the form $\rho^{\otimes n}$, we can always weak Schur sample without loss of generality as our first step, obtaining some $\blambda \vdash n$. Cloning reduces then, very informally, to a problem of specifying how we want to map from the irrep $\blambda$ to all the possible irreps $\mu \vdash n+k$, in such a way as to attain a state close to $\rho^{\otimes (n+k)}$, averaged over the randomness of $\blambda$. But it seems reasonable that, no matter what strategy we decide on, we would always want to preserve the information encoded in $\blambda$ as much as possible, and that we should ``build $\mu$ from $\blambda$'' without discarding any part of the diagram.


Nevertheless, we still must show the bound for every $\mu$, and we now describe our proof technique. Our route is somewhat indirect, and avoids the mixed Schur--Weyl theory \cite{GBO24,Ngu24,Gri25} one would expect to arise from a direct proof, since the transposed factor transforms under the dual representation. Instead of reasoning about the bound directly, we begin by considering the partial-transpose of the bound, to obtain a candidate ``un-transposed integral bound'':
\begin{equation*}
    \int q^d_\lambda(P) \otimes q^d_\mu(P) \cdot \dP \preceq \frac{s_{\lambda \cup \mu}(1^r)}{s_{\lambda \cup \mu}(1^d)} \cdot I_{\calQ^d_\lambda} \otimes I_{\calQ^d_\mu}.
\end{equation*}
It turns out that this bound is relatively easy to show; it follows directly from the Littlewood--Richardson rule, which describes how $\calQ^d_\lambda \otimes \calQ^d_\mu$ decomposes into irreps of $U(d)$. However, this is not useful to us, since $X \preceq Y$ does not imply $X^\Gamma \preceq Y^\Gamma$.\footnote{Here, $\Gamma$ denotes the partial transpose of the first factor.} 

We work around this by proving a stronger result: that the quantity 
\begin{equation*}
    \Delta_{\lambda \mu} \coloneq \frac{s_{\lambda \cup \mu}(1^r)}{s_{\lambda \cup \mu}(1^d)} \cdot I_{\calQ^d_\lambda} \otimes I_{\calQ^d_\mu} - \int q^d_\lambda(P) \otimes q^d_\mu(P) \cdot \dP
\end{equation*}
is not only PSD, but also separable. That is, $\Delta_{\lambda \mu}$ lies in the set $\Sep_+(\calQ^d_\lambda : \calQ^d_\mu)$ consisting of nonnegative combinations of matrices $A \otimes B$, where $A$ acts on $\calQ^d_\lambda$, $B$ acts on $\calQ^d_\mu$, and both $A$ and $B$ are PSD. From here, the partially-transposed integral bound, equivalent to the statement $\Delta_{\lambda \mu}^\Gamma \succeq 0$, really does follow directly. This is because $\Sep_+(\calQ^d_\lambda : \calQ^d_\mu)^\Gamma \subseteq \Sep_+(\calQ^d_\lambda : \calQ^d_\mu)$, since if $A$ is PSD, then so too is $A^T$. 

We prove that $\Delta_{\lambda \mu} \in \Sep_+(\calQ^d_\lambda : \calQ^d_\mu)$ by an induction which simultaneously reduces both $r$ and $d$. The base case is when $r=1$ (and $d \geq r$ is arbitrary). When $r = 1$, $\lambda$ and $\mu$ are one-row diagrams, or empty. Then, one of $\lambda$ or $\mu$ necessarily is a subset of the other, and a simple argument similar to \Cref{eq:tech_overview_4} works. Arguing the inductive step is more involved, and is the most technical part of our argument. The branching rule for $U(d) \big \downarrow_{U(d-1)}$ allows the spaces $\calQ^{d-1}_{\lambda'}$ and $\calQ^{d-1}_{\mu'}$ to be embedded into $\calQ^d_\lambda$ and $\calQ^d_\mu$, for $\lambda' \preceq \lambda$ and $\mu' \preceq \mu$. By embedding an operator on $\calQ^{d-1}_{\lambda'} \otimes \calQ^{d-1}_{\mu'}$ into $\calQ^d_\lambda \otimes \calQ^d_\mu$, and averaging over all conjugations by unitaries, we obtain a ``lifting'' operation which maps $\Sep_+(\calQ^{d-1}_{\lambda'} : \calQ^{d-1}_{\mu'})$ into $\Sep_+(\calQ^d_\lambda : \calQ^d_\mu)$, and whose outputs are unitarily invariant. We then use this lift to derive a recursion expressing $\Delta_{\lambda \mu}$ as a nonnegative linear combination of lifts of $\Delta_{\lambda' \mu'}$, and additional manifestly PSD and separable terms.

\paragraph{A further application: mixed state transposition.} The problem of pure state cloning is closely related to the problem of pure state \emph{transposition} ~\cite{buzek99,Buscemi_2003}. In this problem, one is given $n$ copies of a pure state $\ketbra{u} \in \C^{d \times d}$, and asked to prepare a state which has good fidelity with $\ketbra{u}^{T,\otimes k}$.\footnote{Here, the transpose is to be taken in the computational basis, instead of the Schur basis.} 

A channel that transposes \emph{exactly} is forbidden in quantum mechanics, since it violates complete positivity. Exact cloning is also forbidden, but for a different reason: it violates linearity. Despite this difference, in \cite{BGSQ26}, the authors showed, perhaps surprisingly, that optimal pure state cloning (i.e.\ Werner's cloner) and the optimal pure state transposer are in fact complementary channels, and achieve the same fidelity for all inputs. 

On account of this connection, it seems reasonable that our techniques might also be applicable for analyzing mixed state transposition. In particular, we can define a natural mixed state transposition map by composing the random purification channel with the optimal pure state transposer of \cite{BGSQ26}, and then tracing out the purification registers. We call the resulting map the \emph{PTT transposer}, for its sequence of actions: \emph{purify--transpose--trace}. An almost identical analysis to the cloning upper bound shows that this mixed state transposer can use $n = O(krd/\eps)$ to produce a state with fidelity $1-\eps$ to $\rho^{T,\otimes k}$. 

Again, by considering projector states, we are able to show a matching lower bound of $\Omega(krd/\eps)$. This gives us the following result.

\begin{theorem}[The PTT transposer is sample-optimal] \label{thm:main_result_transposition_intro}
    Suppose $\calT: \states(\C^d)^{\otimes n} \to \states(\C^d)^{\otimes k}$ is a transposition channel such that
    \begin{equation*}
         \Fid \big(  \calT(\rho^{\otimes n}) , \rho^{T,\otimes k}\big) \geq 1-\eps,
    \end{equation*}
    for all states $\rho$ of rank at most $r$. Then we must have $n = \Omega(krd/\eps)$, for $d \geq 2$, and $\epsilon \leq \epsilon_0$, where $\epsilon_0$ is a universal constant. Moreover, the PTT transposer achieves this guarantee with $n = O(krd/\eps)$ copies. Thus, the sample complexity of producing $k$ copies of the transpose of any unknown rank-$r$ input state $\rho^{\otimes n}$ to fidelity $1-\eps$ is $\Theta(krd/\eps)$. 
\end{theorem}


Our lower bound follows the same overall strategy of the cloning lower bound. In particular, we show that rank-$r$ projector states form a family of hard instances. However, there is one major simplification in our argument for transposition: since we aim to produce a state with high fidelity to the \emph{transposed} projector, the analogue of \Cref{eq:tech_overview_2} obtained in this case has a transpose on each factor, and is instead:
\begin{equation*}
\E_{\blambda} \bigg[ \frac{s_{\blambda}(1^d)}{s_{\blambda}(1^r)} \cdot \max_{ \substack{\mu \vdash k \\ \ell(\mu) \leq r}} \norm{ \int_P q_{\blambda}(P)^T \otimes q_\mu(P)^T \cdot \dP }_\infty \bigg] = \E_{\blambda} \bigg[ \frac{s_{\blambda}(1^d)}{s_{\blambda}(1^r)} \cdot \max_{ \substack{\mu \vdash k \\ \ell(\mu) \leq r}} \norm{ \int_P q_{\blambda}(P) \otimes q_\mu(P) \cdot \dP }_\infty \bigg].
\end{equation*}
Analyzing the \emph{un-transposed} integral is much simpler than the partially-transposed integral, and we can give an easier argument that nevertheless leads to the same inequality, \Cref{eq:tech_overview_3}, from before. As for cloning, \Cref{eq:tech_overview_3}  implies our desired lower bound. 

\subsection{Open problems}

Our work leaves open many problems. Here we describe some interesting directions for future work.  

\paragraph{Exact optimality.} 

In this work, we have shown the PCT cloner uses a number of copies $n$ that scales optimally in the parameters $k, d, r, \eps$, as is customary for mixed state algorithms. 
We do not, however, determine the exact optimal fidelity, or find a channel that attains it, as Werner did for the pure state case.
For the mixed state case, such questions of exact optimality are somewhat less canonical than in the pure state case. For pure states, fidelity is essentially the only figure of merit; for mixed states, one has to pick a figure of merit, and thus the exactly optimal protocol may depend on this choice. It is not entirely clear how natural such algorithms might be, since it is also not clear how natural any given figure of merit really is.

Nevertheless, there is an intriguing analogy with state tomography here. Recent work \cite{GMPW26a,GMPW26b} shows that, in a certain information-theoretic sense, Keyl's algorithm is the optimal algorithm for state estimation, rather than the algorithm due to \cite{PSTW25} which applies the random purification channel, and then uses the optimal pure state learning algorithm to learn the purification. The PCT cloner is the direct analogue of the latter construction, but there is no known analogue for the former. Could there be a ``Keyl's algorithm for cloning'': a natural, intrinsically mixed-state cloner which improves upon the PCT cloner at the level of exact fidelity? If so, perhaps this cloner is directly related to Keyl's algorithm, just as optimal pure state cloning and learning are related \cite{Chi11}. Similarly, such a cloner might be related to an exact optimal algorithm for mixed state transposition; in the pure-state case, the optimal algorithms are complementary channels \cite{BGSQ26}.

\paragraph{Information-theoretic lower bounds.} In this work, we give a representation theoretic proof via studying projector cloning. In many ways, our argument is the cloning analogue of the lower bound for learning in fidelity given in \cite{SSW25}. This tomography lower bound was subsequently recovered by Chen, Zhang and Yu using information-theoretic arguments\footnote{Here, we mean ``information-theoretic'' in the sense of using tools from information theory, rather than in the sense of an absence of computational restrictions.} \cite{CZY26}.\footnote{In fact, Chen, Zhang and Yu obtain a more general lower bound on learning channels in diamond distance. When the input dimension is taken to be $1$, this reduces to a lower bound on learning states in trace distance, which also implies a fidelity lower bound.} Such arguments have previously also been used to derive quantitative limitations on cloning \cite{LW17}. Can information-theoretic methods be used to show the optimal $\Omega(krd/\eps)$ cloning lower bound as well?

\paragraph{Broadcasting.} A version of cloning commonly considered in the literature is \emph{broadcasting}, in which we judge the quality of the output not by how close it is to $\rho^{\otimes (n+k)}$, but by how close the one-register marginals are to $\rho$ \cite{BCF+96,DMP05,Chiribella_2006,DF07}. After symmetrizing, this amounts to minimizing $n$ subject to
\begin{equation*}
    \min_{\rho} \Big[ \Fid \big(  \calC(\rho^{\otimes n})_{\reg{1}} , \rho\big)\Big] \geq 1-\eps,
\end{equation*}
In \cite{KW99}, Keyl and Werner showed that Werner's pure state cloner is not only optimal for cloning, but also for pure state broadcasting. We conjecture that the PCT cloner is likewise sample-optimal for mixed state broadcasting.

\paragraph{When is cloning as hard as learning?} 

Our results show that, for arbitrary rank-$r$ inputs, cloning even a single additional copy has the same $\Theta(rd/\eps)$ scaling as learning the input state in fidelity \cite{PSW25,Yue23,SSW25}. The same phenomenon was recently established for stabilizer states, a structured sub-class of pure states, in \cite{BCM26}. However, it is not always true that \emph{any} family of states requires as many samples to clone as to learn. For example, classical probability distributions can be embedded as diagonal mixed states, and for such states cloning requires fewer samples than learning \cite{AGSV24,AGH+24}. It would be interesting to understand when the two complexities are equal, or to give a natural family of genuinely ``quantum'' states for which cloning is asymptotically easier than learning.\footnote{Two families of structured states for which we strongly expect cloning to be as hard as learning are the bosonic and fermionic Gaussian states. On account of similarities between the tomography of Gaussian states and the tomography of general states, we conjecture that $\Theta(m^2/\eps)$ copies are required to clone $m$-mode Gaussian states. We expect that a Gaussian version of the random purification channel \cite{CFF+26} can be composed with a Gaussian version of the Werner cloner to yield an optimal upper bound, and that considering pure Gaussian states should suffice to prove a lower bound, but leave this question to the v2.}



\paragraph{Channel cloning.} 
For many questions about quantum states, there is an analogous question concerning quantum channels. With respect to cloning, we can ask: given access to $n$ uses of an unknown quantum channel $\Phi$, how large does $n$ have to be in order to approximate $n+k$ uses of the channel? Here, it is reasonable to measure success in diamond distance. That is, one seeks a superchannel $\calC$ such that
\begin{equation*}
    \sup_{\Phi} \norm{\calC(\Phi^{\otimes n})-  \Phi^{\otimes (n+k)} }_{\diamond}  \leq \eps.
\end{equation*}
The problem of cloning channels has been considered in the unitary special case \cite{CDP08,D_r_2015,CYH15} and the general channel setting \cite{SGKS25}. However, even for unitary channels, this problem is not entirely understood, though the $1 \to 2$ case was solved in \cite{CDP08} and the asymptotic conversion rate has been established (albeit with respect to the average-case error). This question is particularly timely in light of a recent flurry towards settling the optimal sample complexity of channel tomography \cite{MB25,CYZ25,GMZ+25,YNU+25,CGO+26}. Is channel cloning as hard as channel learning, as we find in the quantum state case? Does channel cloning exhibit a Heisenberg-to-classical phase transition in the sample complexity, as suggested by asymptotic results \cite{SGKS25}?


\subsection{Concurrent and independent work}
In the final stages of preparing this manuscript, we became aware of the concurrent and independent work by \cite{JSO26}. In this work, the authors also proved a sample complexity of $\Theta(krd/\eps)$, using the purify--clone--trace construction for their upper bound, and analyzing projector states for their lower bound. While our upper-bound arguments appear to be essentially the same, our lower-bound proofs appear to take different routes. We additionally prove the optimal sample complexity for mixed state transposition. On the other hand, \cite{JSO26} establish a reduction from $k=1$ cloning to tomography, showing that cloning is no harder than tomography.

%% file: prelims.tex
Some general notation: Throughout, we use \textbf{boldface} for random variables, and \textsf{sans serif} for registers (i.e.\ subsystems of larger quantum systems). For example, if $\rho = \rho_{\reg{A}_1 \dots \reg{A}_n}$ is a state on $n$ registers, where the $i$-th register is labeled $\reg{A}_i$, we write $\rho_{\reg{A}_1} = \tr_{\reg{A}_2\dots \reg{A}_n}(\rho)$ for the reduced state on the first register. We write $[d]$ as shorthand for the set of integers $\{1, 2, \dots, d\}$.

Some quantum information notation and conventions: First, for a Hilbert space $\calH$, we write $\states(\calH)$ for the set of quantum states on that space, i.e.\ the set of positive semi-definite Hermitian operators with unit trace. Next, we take fidelity to be the quantity $\Fid(\rho, \sigma) = \norm{ \sqrt{\rho} \sqrt{\sigma} }_1^2$. In particular, for pure state inputs we have $\Fid(u, v ) = | \braket{u}{v} |^2$. Lastly, in the context of the Choi state, we will take the transpose in the Schur basis, which is described in \Cref{sec:Schur-Weyl_duality}.

\subsection{The symmetric subspace} \label{sec:symmetric_subspace}

In this section, we define the symmetric subspace, and collect a few of its basic properties. For much more, see e.g.\ \cite{Har13} or \cite{Mel24}. 

Denote the symmetric group on $n$ objects by $S_n$. The symmetric group has a natural action $\calP^{(d,n)}$ on the $n$-copy Hilbert space $(\C^d)^{\otimes n}$ by permuting the tensor factors. Formally, $\calP^{(d,n)}$ is defined by its action on standard basis elements as
\begin{equation*}
    \calP^{(d,n)}(\pi) \cdot \ket{i_1, \dots, i_n} \coloneq \ket*{i_{\pi^{-1}(1)}, \dots, i_{\pi^{-1}(n)}}
\end{equation*}
for all $(i_1, \dots, i_n) \in [d]^n$. We will often drop the superscript $(d,n)$ when these parameters are clear from context.

The \emph{symmetric subspace}, denoted $\Sym{d}{n}$, is the subspace of $(\C^d)^{\otimes n}$ which is invariant under all permutations. That is, 
\begin{equation*}
    \Sym{d}{n} \coloneq \mathrm{span}\big\{ \ket{\Psi} \in (\C^d)^{\otimes n} \, : \, \calP(\pi) \cdot \ket{\Psi} = \ket{\Psi}, \, \forall \pi \in S_n\big\}. 
\end{equation*}
Equivalently, it is the subspace spanned by multi-copy pure states:
\begin{equation*}
    \Sym{d}{n} \coloneq \mathrm{span}\big\{ \ket{\psi}^{\otimes n} \in (\C^d)^{\otimes n} \big\}.
\end{equation*}
We will denote the projector onto the symmetric subspace by $\Pisym{d}{n}$, and the dimension of the subspace by $d[n]$. We have
\begin{equation*}
    d[n] \coloneq \dim\big( \Sym{d}{n} \big) = \tr(\Pisym{d}{n}) = \binom{n+d-1}{n}. 
\end{equation*}
We also have the following simple, but useful, bound.
\begin{lemma} \label{lem:ratio_of_symmetric_subspace_dimensions}
    Let $n,d \geq 1$ and $k \geq 0$. The following inequality holds:
    \begin{equation*}
        \frac{d[n]}{d[n+k]} \geq 1 - \frac{kd}{n}. 
    \end{equation*}
\end{lemma}
\begin{proof}
    We have
    \begin{equation*}
        \frac{d[n]}{d[n+k]} = \frac{(n+k)!}{n!} \cdot \frac{(n+d-1)!}{(n+k+d-1)!} = \frac{(n+1)\cdots(n+k)}{(n+d)\cdots(n+d+k-1)} = \prod_{j=1}^{k} \frac{n+j}{n+d-1+j} \geq \Big(\frac{n+1}{n+d}\Big)^k. 
    \end{equation*}
    From here, we use
    \begin{equation*}
        \Big(\frac{n+1}{n+d}\Big)^k = \Big( 1 - \frac{d-1}{n+d} \Big)^k \geq 1 - \frac{k(d-1)}{n+d} \geq 1 - \frac{kd}{n}.
    \end{equation*}
    This completes the proof. 
\end{proof}

\subsection{Werner's pure state cloner}

In \cite{Wer98}, Werner constructed the optimal channel for cloning pure states. Here, we review this construction. \emph{Werner's pure state cloner} is the channel $\Werner^{(d,n,k)}: \states(\Sym{d}{n}) \to \states(\Sym{d}{n+k})$ given by
\begin{equation*} \label{eq:Werners_cloner}
    \Werner^{(d,n,k)}(\rho) \coloneq \frac{d[n]}{d[n+k]} \cdot \Pisym{d}{n+k} \cdot \big( \rho \otimes I^{\otimes k} \big) \cdot \Pisym{d}{n+k}
\end{equation*}
for any $\rho \in \states(\Sym{d}{n})$. 
This construction is clearly completely positive; trace preservation is less obvious, but also holds.\footnote{One way to see trace preservation is to note that Werner's cloner is actually an example of a \emph{Petz recovery map}. In this case, it is Petz recovery from the channel which traces out $k$ registers, with respect to the state $\sigma = \E \ketbra{\bu}^{\otimes (n+k)} = \Pisym{d}{n+k}/d[n+k]$. Petz recovery maps are trace-preserving on the support of $\Tr_{[n+k]\setminus [k]}\sigma = \Pisym{d}{n}/d[n]$, i.e.\ on the support of the symmetric subspace..}

\begin{lemma}[Werner's cloner is optimal \cite{Wer98}] \label{lem:Werner_optimality}
    For any input state $\ket{u} \in \C^d$, we have
    \begin{equation*}
        \Fid\Big( \Werner\big(\ketbra{u}^{\otimes n} \big), \ketbra{u}^{\otimes (n+k)} \Big) = \frac{d[n]}{d[n+k]}.
    \end{equation*}
    Moreover, this is the optimal worst-case value. That is, for any pure state cloning channel $\calC$, we have
    \begin{equation*}
        \min_{\ket{u}} \Big[ \Fid\Big( \calC\big(\ketbra{u}^{\otimes n} \big), \ketbra{u}^{\otimes (n+k)} \Big)\Big] \leq \frac{d[n]}{d[n+k]}. 
    \end{equation*}
\end{lemma}

We now give a simple analysis of its sample complexity. 

\begin{lemma}[Sample complexity of pure state cloning]
    With $n = O(kd/\eps)$ copies of a pure state $\ket{u} \in \C^d$, Werner's cloner $\Werner^{(d,n,k)}$ produces an $(n+k)$-copy state with fidelity at least $1-\eps$ with $\ket{u}^{\otimes (n+k)}$. \label{lem:sample_complexity_Werner}
\end{lemma}

\begin{proof}
On any input, Werner's cloner achieves fidelity $d[n]/d[n+k]$, by \Cref{lem:Werner_optimality}. Thus, 
\begin{equation*}
    \min_{\ket{u}} \Big[ \Fid\Big( \Werner\big(\ketbra{u}^{\otimes n} \big), \ketbra{u}^{\otimes (n+k)} \Big)\Big] = \frac{d[n]}{d[n+k]} \geq 1 - \frac{kd}{n},
\end{equation*}
using \Cref{lem:ratio_of_symmetric_subspace_dimensions}. If we take $n=\lceil kd/\eps \rceil$, this is at least $1- \eps$. 
\end{proof}

\subsection{The random purification channel}

The random purification channel was introduced in \cite{TWZ25} and has since proved to be an extremely useful tool in quantum information. Informally, the channel acts by mapping copies of a mixed state to copies of a consistent, but random, purification of that mixed state. 

More formally, Tang, Wright, and Zhandry showed that there exists a channel $\Purify^{(d,r, n)}: \states(\C^d)^{\otimes n} \to \states(\C^d \otimes \C^r)^{\otimes n}$ which acts as follows. Given any input of the form $\rho^{\otimes n}$, where $\rho \in \states(\C^d)$ is any state of rank at most $r$, we have
\begin{equation*}
    \Purify^{(d,r,n)}(\rho^{\otimes n}) = \E_{\ket{\brho}} \big[ \ketbra{\brho}^{\otimes n} \big].
\end{equation*}
Here, $\ket{\brho} \in \C^d \otimes \C^r$ is a \emph{random purification} of $\rho$. By \emph{random}, we mean a state of the form $\ket{\brho} = (I \otimes \bU) \cdot \ket{\rho_0}$, where $\ket{\rho_0}$ is any fixed purification of $\rho$, and $\bU$ is Haar random. Thus, the action is equivalently described as
\begin{equation*}
    \Purify^{(d,r,n)}(\rho^{\otimes n}) = \E_{\bU \sim \mathrm{Haar}} \big[ (I \otimes \bU)^{\otimes n} \cdot \ketbra{\rho_0}^{\otimes n} \cdot (I \otimes \bU^\dagger)^{\otimes n} \big].
\end{equation*}
We will also denote the original $d$-dimensional registers as $\reg{A}_1, \dots, \reg{A}_n$, and the collection of all such registers as $\reg{A}$. Similarly, the auxiliary $r$-dimensional registers will be denoted $\reg{B}_1, \dots, \reg{B}_n$, and collectively as $\reg{B}$. 

\subsection{Representation theory}

In this section, we give a minimal review of the representation theory we need. For much more on representation theory for quantum information, see e.g.\ \cite{Har05,Wri16,Gri25}.

\subsubsection{Young diagrams and tableaux}

A \emph{partition} of $n$ is a tuple $\lambda = (\lambda_1, \dots, \lambda_d)$ such that $\sum_{i=1}^d \lambda_i = n$, and $\lambda_1 \geq \dots \geq \lambda_d \geq 0$. We write $\Par(n,d)$ for the set of all partitions of size $n$ into $d$ parts. Sometimes we will want to fix only either $n$ or $d$, in which case we will write $\Par(n,-)$ or $\Par(-,d)$. In the former case, we identify partitions that differ only by trailing zeroes, e.g.\ $(1)$ and $(1,0)$ refer to the same partition, when viewed as elements of $\Par(1,-)$. We will sometimes also write $\ell(\lambda)$ for the \emph{length} of the partition $\lambda$, which is the number of nonzero entries in $\lambda$, and $\lambda \vdash n$ to denote that $\lambda$ is a partition of $n$. 

Partitions can be represented diagrammatically with \emph{Young diagrams}. Young diagrams consist of stacked rows of left-justified boxes (or \emph{cells}), such that the $i$-th row contains $\lambda_i$ boxes. 

\begin{figure}[h!]
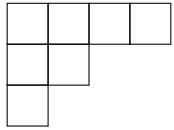

\centering
\ytableausetup{boxframe=normal}

\begin{ytableau}
\, & \, & \, & \, \\
\, & \, \\
\,
\end{ytableau}

\caption{The Young diagram corresponding to $\lambda = (4,2,1)$.}
\label{fig:Young_fig_1}
\end{figure}

We can also think of a Young diagram as a subset of all unit cells in the plane, identifying the box in row $i$ and column $j$ with the cell at position $(i,j)$. With this viewpoint, it is natural to index boxes in a diagram $\lambda$ with e.g.\ $\Box \in \lambda$ or $(i,j) \in \lambda$, or to consider set operations on Young diagrams, such as $\lambda \cup \mu$, $\lambda \cap \mu$, or $\lambda \setminus \mu$, for two Young diagrams $\lambda$ and $\mu$. Note that while the union and intersection of diagrams is also a Young diagram, the set difference is not. 

\begin{figure}[h!]
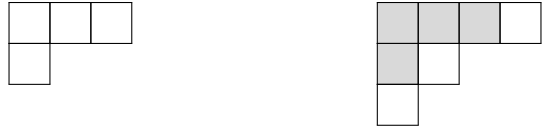

\centering
\ytableausetup{boxframe=normal}

\begin{ytableau}
\, & \, & \,  \\
\, 
\end{ytableau}
\hspace{3cm}
\begin{ytableau}
*(gray!30)~ & *(gray!30)~ & *(gray!30)~ & \, \\
*(gray!30)~ & \, \\
\,
\end{ytableau}

\caption{On the left, $\mu = (3,1)$. On the right, $\lambda = (4,2,1)$. We have $\mu \subseteq \lambda$ -- the boxes of $\mu$ are shaded inside $\lambda$.}
\label{fig:Young_fig_2}
\end{figure}

We will use the following definitions and conventions.
\begin{itemize}
    \item The \emph{content} of a cell, which is the number $c(i,j) \coloneq j-i$. See \Cref{fig:Young_fig_3}.
    \item the \emph{hook-length} of a cell in a diagram $\lambda$, denoted $h_\lambda(i,j)$, which counts the number of cells either vertically below $(i,j)$, or horizontally to the right of $(i,j)$, including $(i,j)$ itself. See \Cref{fig:Young_fig_3}.
    \item If a diagram $\mu$ can be obtained by adding a box in row $i$ to $\lambda$, we write e.g.\ $\mu = \lambda + e_i$. Similarly, if we add $t$ boxes to the $i$-th row, we write $\lambda+t \cdot e_i$ for the resulting diagram. See \Cref{fig:Young_fig_3}.
    \item For two Young diagrams $\lambda, \lambda' \in \Par(-,d)$, such that for each $i$, $\lambda_{i+1} \leq \lambda'_i \leq \lambda_i$ (taking $\lambda_{d+1} = 0$), we say that $\lambda'$ \emph{interlaces} $\lambda$, and write $\lambda' \preceq \lambda$. For example, in \Cref{fig:Young_fig_2}, $\mu \preceq \lambda$. 
\end{itemize}

\begin{figure}[h!]
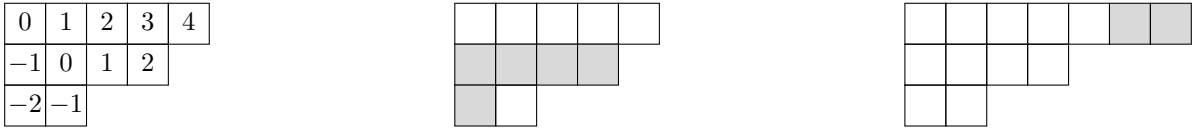

\centering
\ytableausetup{boxframe=normal}

\begin{ytableau}
0 & 1 & 2 & 3 & 4\\
-1 & 0  & 1 & 2\\
-2 & -1
\end{ytableau}
\hspace{3cm}
\begin{ytableau}
\, & \, & \, & \, & \, \\
*(gray!30)~ & *(gray!30)~  & *(gray!30)~ & *(gray!30)~\\
*(gray!30)~ & \,
\end{ytableau}
\hspace{3cm}
\begin{ytableau}
\, & \, & \, & \, & \, & *(gray!30)~ & *(gray!30)~ \\
\, & \,  & \, & \,\\
\, & \,
\end{ytableau}

\caption{On the left, the diagram $\sigma = (5,4,2)$, with each cell filled with its content. In the middle, $\sigma$ with the cells counted by the hook-length of $(2,1)$ shaded. Thus, $h_\sigma(2,1) = 5$. On the right, the diagram $\sigma + 2 \cdot e_1$. The new boxes are shaded.}
\label{fig:Young_fig_3}
\end{figure}


A Young tableau is a Young diagram with a symbol labeling each box. A \emph{standard Young tableau} (SYT) is a Young tableau of shape $\lambda \vdash n$ where the boxes are labeled with unique symbols in $[n]$, such that the labels strictly increase to the right along rows, and down columns. A \emph{semistandard Young tableau} (SSYT) with \emph{alphabet $[d]$} is a Young tableau where the boxes are labeled with symbols in $[d]$, such that labels are weakly increasing to the right along rows, and strictly increasing down columns.

\begin{figure}[h!]
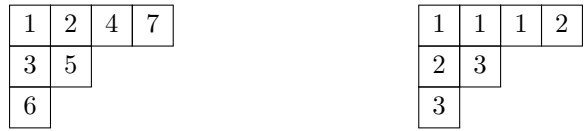

\centering
\ytableausetup{boxframe=normal}

\begin{ytableau}
1 & 2 & 4 & 7 \\
3 & 5 \\
6
\end{ytableau}
\hspace{3cm}
\begin{ytableau}
1 & 1 & 1 & 2 \\
2 & 3 \\
3
\end{ytableau}

\caption{On the left, an SYT $S$ of shape $\tau = (4,2,1)$. On the right, an SSYT $T$ of shape $\tau$ and alphabet $[3]$.}
\label{fig:YD_fig_5}
\end{figure}

For an SSYT $T$ of shape $\lambda$ and alphabet $[d]$, we write $\lambda = \sh(T)$. We can also construct the SSYT $T_{\leq t}$ which consists of keeping only those boxes of $T$ which contain symbols at most $t$. A key property of SSYTs is that $\sh(T_{\leq (t-1)}) \preceq \sh(T_{\leq t})$ for all $t$. 

\begin{figure}[h!]
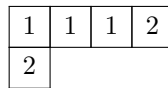

\centering
\ytableausetup{boxframe=normal}

\begin{ytableau}
1 & 1 & 1 & 2 \\
2 
\end{ytableau}

\caption{The SSYT $T_{\leq 2}$ obtained from $T$ in \Cref{fig:YD_fig_5}. We also have $\sh(T_{\leq 2}) = (3,1)$.}
\end{figure}

\subsubsection{Irreps of the symmetric and unitary groups}

We denote the symmetric group on $n$ objects by $S_n$. The irreducible representations of the symmetric group are labeled by Young diagrams $\lambda$ with $\lambda \in \Par(n, -)$. The irrep space corresponding to $\lambda$ is written $\calP_\lambda$, and the action is denoted $p_\lambda$. There is a basis of $\calP_\lambda$ with elements labeled by SYTs of shape $\lambda$.

We denote the $d$-dimensional unitary group by $U(d)$. The polynomial irreducible representations of the unitary group are labeled by Young diagrams $\lambda$ with $\ell(\lambda) \leq d$, called the \emph{Gelfand--Tsetlin basis}. The irrep space corresponding to $\lambda$ is written $\calQ^d_\lambda$, and the action is denoted $q^d_\lambda$. There is a basis of $\calQ^d_\lambda$ with elements labeled by SSYTs of shape $\lambda$ and alphabet $[d]$. There is a formula for the dimension of this space --- equivalently, for $|\SSYT(\lambda,d)|$ --- called the \emph{hook-content formula} \cite[Theorem 7.21.2]{Sta99}:
\begin{equation}
        \dim(\calQ^d_\lambda) = \prod_{ \Box \in \lambda } \frac{d + c(\Box)}{h_\lambda(\Box)}. \label{eq:hook_content_formula}
\end{equation}
The character corresponding to $\lambda$ is 
\begin{equation*}
    \tr( q^d_\lambda(U) ) = s_\lambda(\alpha_1, \dots, \alpha_d),
\end{equation*}
where $s_\lambda$ is the \emph{Schur polynomial} corresponding to $\lambda$, and $(\alpha_1, \dots, \alpha_d)$ are the eigenvalues of $U$. Another expression for $s_\lambda$ is
\begin{equation*}
    s_\lambda(\alpha_1, \dots, \alpha_d) = \sum_{T \in \SSYT(\lambda,d)} \alpha_1^{w_1(T)} \cdots \alpha_d^{w_d(T)},
\end{equation*}
where $w_i(T)$ is the number of $i$'s in $T$. Accordingly, if $\alpha_i \geq 0$ for all $i$, then $s_\lambda(\alpha)$ is positive unless there are no such $T$. 
In particular, $s_\lambda(1^d) = \dim(\calQ^d_\lambda)$. For any $t \in [d]$, we will write $s_\lambda(1^t)$ as shorthand for $s_\lambda(1^t, 0^{d-t})$, and note that this quantity is equal to $\dim(\calQ^t_\lambda)$. The following is a useful fact. 

\begin{lemma} \label{lem:ratio_of_s's_subset_diagrams}
    Let $\sigma \in \Par(-,d)$ and $\tau \in \Par(-,d)$, such that $\sigma \subseteq \tau$. Then for any $t \in [d]$,
    \begin{equation*}
        \frac{s_\sigma(1^t)}{s_\sigma(1^d)} \geq \frac{s_\tau(1^t)}{s_\tau(1^d)}. 
    \end{equation*}
\end{lemma}

\begin{proof}
By the hook-content formula $s_\lambda(1^t) = 0$ iff $\ell(\lambda) > t$. Indeed, there are no SSYTs of shape $\lambda$ with alphabet $[t]$, since in particular the first column, which would contain more than $t$ symbols, has to be an increasing subsequence of $[t]$. If $\ell(\lambda) \leq t$, there is at least one such SSYT, obtained by filling the $i$-th row with symbol $i$, and so $s_{\lambda}(1^t) > 0$. 

Now, if $\ell(\tau) > t$, then we have nothing to prove, as the right-hand side of the desired inequality is zero, and the left-hand side is nonnegative. So assume $\ell(\tau) \leq t$. Using the hook-content formula, \Cref{eq:hook_content_formula}, we have 
     \begin{equation*}
        \frac{s_\tau(1^{t})}{s_\tau(1^{d})} \cdot  \frac{s_\sigma(1^{d})}{s_\sigma(1^{t})} = \bigg( \prod_{\Box \in \tau} \frac{t + c(\Box)}{d+c(\Box)} \bigg) \cdot \bigg( \prod_{\Box \in \sigma} \frac{d + c(\Box)}{t+c(\Box)} \bigg) = \prod_{\Box \in \tau \setminus \sigma} \frac{t + c(\Box)}{d+c(\Box)},
    \end{equation*}
    However, each term in the right-most product is in $[0,1]$ for all $\Box \in \tau$, since $c(\Box) > -\ell(\tau)$, and since $t \leq d$. Thus the product is at most $1$, and we conclude
    \begin{equation*}
         \frac{s_\sigma(1^{t})}{s_\sigma(1^{d})} \geq \frac{s_\tau(1^{t})}{s_\tau(1^{d})}.
    \end{equation*} 
    This completes the proof.
\end{proof}

In our proof of the lower bound, we will need to know a little about the branching rule for $U(d)$. View $U(d-1) \subseteq U(d)$ by letting $V \in U(d-1)$ act nontrivially only on the first $d-1$ elements. Then the representation $(\calQ^d_\lambda, q^d_\lambda)$ restricts to $U(d-1)$ as
\begin{equation*}
    \calQ^d_\lambda \big \downarrow_{U(d-1)} \cong \bigoplus_{\substack{\lambda' \in \Par(-,d-1) \\ \lambda' \preceq \lambda}} \calQ^{d-1}_{\lambda'}. 
\end{equation*}
This branching rule is key to the construction of the Gelfand--Tsetlin basis, whose basis elements are indexed by SSYTs of shape $\lambda$. The vectors appearing in the block corresponding to $\lambda'$ are those SSYTs such that $\sh(T_{\leq (d-1)}) = \lambda'$.


\subsubsection{Schur--Weyl duality} \label{sec:Schur-Weyl_duality}

The symmetric group, $S_n$, and the unitary group, $U(d)$, have natural unitary actions on the space $(\C^d)^{\otimes n}$. We already described this natural action for $S_n$ in \Cref{sec:symmetric_subspace}. It is the action $\calP^{(d,n)}$ given by
\begin{equation*}
    \calP^{(d,n)}(\pi) \cdot \ket{i_1, \dots, i_n} \coloneq \ket*{i_{\pi^{-1}(1)}, \dots, i_{\pi^{-1}(n)} },
\end{equation*}
for all permutations $\pi \in S_n$, and all basis elements $(i_1, \dots, i_n) \in [d]^n$. The natural action of the unitary group is the action $\calQ^{(d,n)}$ with
\begin{equation*}
    \calQ^{(d,n)}(U) \cdot \ket{i_1, \dots, i_n} \coloneq U \ket{i_1} \otimes \dots \otimes U \ket{i_n}, 
\end{equation*}
for all unitaries $U \in U(d)$, and again all basis elements. We will drop the superscripts $(d,n)$ when these are clear from context. 

These two actions commute, and so we obtain a representation of $S_n \times U(d)$ from their product. A result known as \emph{Schur--Weyl duality} tells us how this product action decomposes into the irreps of the two constituent groups. It says:
\begin{equation*}
    (\C^d)^{\otimes n} \cong \bigoplus_{\lambda \in \Par(n,d)} \calP_\lambda \otimes \calQ_\lambda^d.
\end{equation*}
In particular, there exists a unitary change-of-basis on $(\C^d)^{\otimes n}$ which implements this isomorphism. We call such a unitary the \emph{Schur transform} and denote it $\Schur^{(d,n)}$. That is, for all $\pi \in S_n$ and $U \in U(d)$, we have
\begin{equation*}
    \Schur^{(d,n)} \cdot \Big( \calP^{(d,n)}(\pi) \cdot \calQ^{(d,n)}(U) \Big) \cdot \Schur^{(d,n), \dagger} =  \bigoplus_{\substack{\lambda \vdash n \\ \ell(\lambda) \leq d}} p_\lambda(\pi) \otimes q^d_\lambda(U).
\end{equation*}
The basis of $(\C^d)^{\otimes n}$ on the right-hand side is known as the \emph{Schur basis}.

These unitary representations extend from $U(d)$ to all operators on $\C^d$. Viewing $\rho^{\otimes n}$ as $\calQ^{(d,n)}(\rho)$, we obtain
\begin{equation*}
    \Schur^{(d,n)} \cdot \rho^{\otimes n} \cdot \Schur^{(d,n), \dagger} =  \bigoplus_{\substack{\lambda \vdash n \\ \ell(\lambda) \leq d}} I_{\calP_\lambda} \otimes q^d_\lambda(\rho).
\end{equation*}
Thus, all states of the form $\rho^{\otimes n}$ are simultaneously block-diagonalized in the Schur basis. Moreover, any such state is thus a \emph{classical mixture} over states of the form 
\begin{equation*}
    \frac{I_{\calP_{\lambda}}}{\dim(\calP_{\lambda})} \otimes \frac{q^d_\lambda(\rho)}{s_\lambda(\rho)},
\end{equation*}
supported in the individual $\lambda$ blocks.
Here, $s_\lambda(\rho)$ is $s_\lambda$ evaluated at the eigenvalues of $\rho$. The measurement which projects into one of the blocks is called \emph{weak Schur sampling}, and the probability that we obtain block $\lambda$ is $\dim(\calP_{\lambda}) \cdot s_\lambda(\rho)$.

%% file: cloning.tex

\subsection{Upper bounds via the random purification channel}

Combining the random purification channel with Werner's pure state cloner straightforwardly gives a simple mixed state cloner. 

{
\floatstyle{boxed} 
\restylefloat{figure}
\begin{figure}[H]
Given $n$ copies of a rank-$r$ mixed state $\rho_{\reg{A}}$:
\begin{enumerate}
    \item Apply $\Purify^{(d,r,n)}$ to prepare $n$ copies of a random purification $\ket{\brho}_{\reg{AB}} \in \C^d \otimes \C^r$. Set $D = d \cdot r$. 
    \item Apply $\Werner^{(D, n, k)}$, yielding a mixed state $\bsigma_{\reg{AB}} \in \states( \C^d \otimes \C^r)^{\otimes (n+k)}$. 
    \item Trace out the auxiliary registers, and output $\Tr_{\reg{B}}(\bsigma_{\reg{AB}}) \in \states(\C^d)^{\otimes (n+k)}$. 
\end{enumerate}
\caption{A mixed state cloning channel, which we call the \emph{PCT cloner}, due to its action: \emph{purify--clone--trace}. We denote the PCT cloner by $\PCT^{(d,n,k)}$. As usual, we will sometimes drop the superscript's parameters when these are clear from context.}
\label{fig:reduction}
\end{figure}
}

This natural construction has already appeared in the literature, in e.g.\ \cite{LTHC26}. Here we give an analysis of its sample complexity. This result, when combined with our lower bounds, shows that the PCT cloner is sample-optimal, despite its simplicity. 

\begin{proposition}[Sample complexity of the PCT cloner] \label{prop:upper_bound_PCT}
    With $n = O(krd/\eps)$ copies of a mixed state $\rho \in \states(\C^d)$, the PCT cloner $\PCT^{(d,n,k)}$ produces an $(n+k)$-copy state with fidelity at least $1-\eps$ with $\rho^{\otimes (n+k)}$. 
\end{proposition}

\begin{proof}
    We start by writing
    \begin{equation*}
        \PCT(\rho^{\otimes n}) = \Tr_{\reg{B}} \Big(\Werner  \big(\Purify \big(\rho^{\otimes n}\big) \big) \Big) = \Tr_{\reg{B}} \Big(\Werner  \Big(\E_{\ket{\brho}} \big[\ketbra{\brho}^{\otimes n}\big] \Big) \Big) = \E_{\ket{\brho}} \Big[ \Tr_{\reg{B}} \Big( \Werner\big(\ketbra{\brho}^{\otimes n} \big) \Big) \Big].
    \end{equation*}
    By concavity of the ``square-root fidelity'' \cite[Corollary 3.26]{Wat18}, we then have
    \begin{equation*}
        \sqrt{\Fid \Big(  \PCT(\rho^{\otimes n}) , \rho^{\otimes (n+k)}\Big)} \geq \E_{\ket{\brho}}\Bigg[ \sqrt{\Fid \Big( \Tr_{\reg{B}} \Big( \Werner\big(\ketbra{\brho}^{\otimes n} \big) \Big), \rho^{\otimes (n+k)} \Big)} \Bigg].
    \end{equation*}
    Next, by data processing of fidelity \cite[Theorem 3.27]{Wat18}, 
    \begin{align*}
        \Fid \Big( \Tr_{\reg{B}} \Big( \Werner\big(\ketbra{\brho}^{\otimes n} \big) \Big), \rho^{\otimes (n+k)} \Big) & = \Fid \Big( \Tr_{\reg{B}} \Big( \Werner\big(\ketbra{\brho}^{\otimes n} \big) \Big), \Tr_{\reg{B}}\big(\ketbra{\brho}^{\otimes (n+k)}\big) \Big) \\
        & \geq \Fid \Big( \Werner\big(\ketbra{\brho}^{\otimes n} \big) , \ketbra{\brho}^{\otimes (n+k)}\Big).
    \end{align*}
    Then, by \Cref{lem:Werner_optimality,lem:ratio_of_symmetric_subspace_dimensions}, we have
    \begin{equation*}
        \Fid \Big( \Werner\big(\ketbra{\brho}^{\otimes n} \big) , \ketbra{\brho}^{\otimes (n+k)}\Big) = \frac{D[n]}{D[n+k]} \geq 1 - \frac{kD}{n} = 1 - \frac{krd}{n}. 
    \end{equation*}
    Since this holds for all $\ket{\brho}$, we obtain
    \begin{equation*}
        \Fid \Big(  \PCT(\rho^{\otimes n}) , \rho^{\otimes (n+k)}\Big) \geq \Bigg( \E_{\ket{\brho}} \Bigg[ \sqrt{\Fid \Big( \Werner\big(\ketbra{\brho}^{\otimes n} \big) , \ketbra{\brho}^{\otimes (n+k)}\Big)} \Bigg] \Bigg)^2 \geq 1 - \frac{krd}{n}. 
    \end{equation*}
    Taking $n = \lceil krd/\eps \rceil$ makes this at least $1 - \eps$. This completes the proof.
\end{proof}

\subsection{Lower bounds via projector cloning}

In this section we prove the following result, which shows that rank-$r$ projector states are among the hardest rank-$r$ states to clone. Throughout, $P \in \C^{d \times d}$ will be a rank-$r$ projector, and $\rho = P/r$ the corresponding projector state. We will continue to use the shorthand $m = n+k$.

\begin{proposition}[A lower bound for cloning projector states; \Cref{prop:projector_lower_bound_tech_overview}, restated]\label{prop:projector_lower_bound}
Any channel $\calC: \states(\C^d)^{\otimes n} \to \states(\C^d)^{\otimes m}$ which clones rank-$r$ projector states to fidelity $1-\eps$ requires at least $n = \Omega(krd/\epsilon)$ copies as input, for $d \geq 2$, $r \leq d/2$, and $\epsilon \leq 1/16$. 
\end{proposition}

We note that restrictions of the form $d \geq 2$, $r < d$ and $\epsilon < 1$ are necessary to avoid cases where you can easily clone for trivial reasons: $d=1$ or $r = d$ or $\epsilon = 1$. We have not tried to optimize either of the restrictions on $r$ or $\eps$, as these will suffice for our purposes.

The remainder of the section is devoted to this statement's proof. Our argument is organized into five main steps.

\paragraph{Step 1.} We first reduce from fidelity, which is nonlinear, to a linear figure of merit, which will be easier to work with. For a quantum state $\sigma$ and a projector $\Pi$, we define the \emph{overlap} as the quantity $\tr(\sigma \cdot \Pi )$. The overlap is a simple quantity, and its utility lies in the following fact. 

\begin{lemma}\label{lem:reduction_from_fidelity_to_overlap}
    Let $\Pi \in \C^{d \times d}$ be a projector of rank $R$, and $\sigma \in \C^{d \times d}$ a quantum state. Then we have 
    \begin{equation*}
        \Fid( \sigma, \Pi/R) \leq \tr( \sigma \cdot \Pi).
    \end{equation*}
\end{lemma}

\begin{proof}
    Let $\{\lambda_i\}_{i \in [d]}$ be the eigenvalues of the Hermitian matrix $\Pi \cdot \sigma \cdot \Pi$, ordered decreasingly. Since $\Pi$ has rank $R$, at most $R$ of these eigenvalues are nonzero. Then by Cauchy-Schwarz,
    \begin{align*}
        \Fid(\sigma, \Pi/R) = \tr( \sqrt{\sqrt{ \Pi/R} \cdot \sigma \cdot \sqrt{\Pi/R}} )^2 = \frac{1}{R} \tr( \sqrt{\Pi \cdot \sigma \cdot \Pi})^2 = \frac{1}{R} \Big( \sum_{i=1}^R \sqrt{\lambda_i}\Big)^2 \leq \sum_{i=1}^{R} \lambda_i  & = \tr( \Pi \cdot \sigma \cdot \Pi) \\ & = \tr(\sigma \cdot \Pi). \qedhere
    \end{align*}
\end{proof}

With cloning, we seek to maximize the fidelity between the projector state $\rho^{\otimes m}$ and the output of our channel $\calC(\rho^{\otimes n})$. So, setting $\Pi \leftarrow P^{\otimes m}$ and $\sigma \leftarrow \calC(\rho^{\otimes n})$, we find 
\begin{equation*}
    \Fid \Big( \calC( \rho^{\otimes n}), \rho^{\otimes m} \Big) \leq \tr\Big( \calC (\rho^{\otimes n} ) \cdot P^{\otimes m} \Big).
\end{equation*}
In particular, if we want to prove that $\Fid (\calC( \rho^{\otimes n}), \rho^{\otimes m} ) \geq 1-\eps$ (for all $\rho$) implies a lower bound on $n$, it suffices to show instead that $\tr(\calC (\rho^{\otimes n} ) \cdot P^{\otimes m} ) \geq 1-\eps$ (for all $\rho$) implies a lower bound on $n$. 

\paragraph{Step 2.} Our next step is to reduce from generic channels to \emph{permutation-invariant} channels. This restriction is without loss of generality, due to the following lemma, which shows that for any channel $\calC$, there exists a permutation-invariant channel $\overline{\calC}$ which performs the same on any input $\rho^{\otimes n}$. 

\begin{lemma} \label{lem:reduction_to_perm_invariant}
    Let $\calC: \states(\C^d)^{\otimes n} \to \states(\C^d)^{\otimes m}$ be any channel. There exists a permutationally invariant channel $\overline{\calC}$ such that for any $P$,
    \begin{equation*}
        \tr\Big( \overline{\calC}(\rho^{\otimes n}) \cdot P^{\otimes m} \Big) = \tr\Big( \calC(\rho^{\otimes n}) \cdot P^{\otimes m} \Big).
    \end{equation*}
    In particular, $\calC$ and $\overline{\calC}$ have the same minimum overlap, and the same average overlap.
\end{lemma}

\begin{proof}
    To obtain $\overline{\calC}$ from $\calC$, conjugate the input and output by random permutations of the appropriate size. That is:
    \begin{equation*}
        \overline{\calC}(X) \coloneq \frac{1}{m! \cdot n!} \sum_{\substack{\pi \in S_{m}\\ \sigma \in S_n}} \calP(\pi) \cdot \Big( \calC \big( \calP(\sigma) \cdot X \cdot \calP(\sigma)^\dagger \big)\Big) \cdot \calP(\pi)^\dagger.
    \end{equation*}
    Note that any permutation of the input and output registers can be absorbed into $\sigma$ and $\pi$ respectively, so that $\overline{\calC}$ is permutation-invariant. For inputs of the form $X = \rho^{\otimes n}$, we have
        $\calP(\sigma) \cdot \rho^{\otimes n} \cdot \calP(\sigma)^\dagger = \rho^{\otimes n}, $
    so that $\overline{\calC}(\rho^{\otimes n})$ simplifies:
    \begin{equation*}
        \overline{\calC}(\rho^{\otimes n}) =\frac{1}{m! } \sum_{\pi \in S_{m}} \calP(\pi) \cdot  \calC ( \rho^{\otimes n} )\cdot \calP(\pi)^\dagger.
    \end{equation*}
    However, since 
        $\calP(\pi)^\dagger \cdot P^{\otimes m} \cdot \calP(\pi) = P^{\otimes m}$
    as well, we have 
    \begin{equation*}
        \tr\Big( \overline{\calC}(\rho^{\otimes n}) \cdot P^{\otimes m} \Big) = \frac{1}{m!} \sum_{\pi \in S_m} \tr\Big( \calC(\rho^{\otimes n}) \cdot \big(\calP(\pi)^\dagger \cdot P^{\otimes m} \cdot \calP(\pi)\big) \Big) = \tr\Big( \calC(\rho^{\otimes n}) \cdot P^{\otimes m} \Big).
    \end{equation*}
    This completes the proof. 
\end{proof}

Therefore, the best possible worst-case overlap is attained by a channel which is permutational-invariant. So, we can assume an overlap-optimal cloner has this property.

\paragraph{Step 3.} We now show that the cloning performance of a permutation-invariant channel can be upper bounded cleanly.

\begin{lemma} \label{lem:symmetric_channels_perfomance_upper_bounded}
    Let $\calC: \states(\C^d)^{\otimes n} \to \states(\C^d)^{\otimes m}$ be any permutation-invariant channel. Then we have
    \begin{equation*}
        \min_P \Big[ \tr\Big( \calC(\rho^{\otimes n}) \cdot P^{\otimes m} \Big) \Big] \leq \sum_{\substack{\lambda \vdash n \\ \ell(\lambda) \leq r}} \frac{\dim(\calP_\lambda) \cdot \dim(\calQ^d_\lambda)}{r^n} \cdot \max_{ \substack{\mu \vdash m \\ \ell(\mu) \leq r}} \norm{ \int_P q_\lambda(P)^T \otimes q_\mu(P) \cdot \dP }_\infty.
    \end{equation*}
\end{lemma}

\begin{proof}
We begin by passing to Choi-Jamiołkowski (CJ) states:
\begin{equation*}
    \tr \Big( \calC (\rho^{\otimes n}) \cdot P^{\otimes m} \Big) = \tr\Big( J({\calC}) \cdot \big(\rho^{T, \otimes n} \otimes P^{\otimes m} \big)\Big) = \frac{1}{r^n} \cdot \tr\Big( J({\calC}) \cdot \big(P^{T, \otimes n} \otimes P^{\otimes m}\big) \Big),
\end{equation*}
where $J({\calC})$ is the CJ state corresponding to $\calC$. Now, since $\calC$'s \emph{worst-case overlap} is at most its \emph{average-case} overlap,
\begin{equation}
    \min_P \Big[ \tr\Big( \calC(\rho^{\otimes n}) \cdot P^{\otimes m} \Big) \Big] \leq \frac{1}{r^n} \cdot \tr( J({\calC}) \cdot \int_P P^{T,\otimes n} \otimes P^{\otimes m} \cdot \dP ). \label{eq:trace_choi_projector_integral}
\end{equation}
Here, $\dP$ is the measure on projectors induced by the Haar measure. We will rewrite the final trace in the Schur basis, starting with
\begin{equation*}
    \Big( \Schur^{(d,n),*}  \otimes \Schur^{(d,m)}\Big) \cdot \Big( P^{T,\otimes n} \otimes P^{\otimes m} \Big) \cdot \Big( \Schur^{(d,n),*}  \otimes \Schur^{(d,m)}\Big)^\dagger = \bigoplus_{\substack{\lambda \vdash n \\ \ell(\lambda) \leq r}}\bigoplus_{\substack{\mu \vdash m \\ \ell(\mu) \leq r}} I_{\calP_\lambda} \otimes I_{\calP_\mu} \otimes q_\lambda(P)^T \otimes q_\mu(P).
\end{equation*}
In the same basis, since $\calC$ is permutation invariant, $J({\calC})$ commutes with permutations on both the input and output registers, and is therefore proportional to the identity on the permutation registers by Schur's lemma. Thus, we have
\begin{equation*}
    \Big( \Schur^{(d,n),*}  \otimes \Schur^{(d,m)}\Big) \cdot  J(\calC) \cdot \Big( \Schur^{(d,n),*}  \otimes \Schur^{(d,m)}\Big)^\dagger = \bigoplus_{\substack{\lambda \vdash n \\ \mu \vdash m}} I_{\calP_\lambda} \otimes \frac{I_{\calP_\mu}}{\dim(\calP_\mu)} \otimes J(\calC)_{\lambda \mu},
\end{equation*}
with $J(\calC)_{\lambda \mu}$ some PSD matrix on $\calQ^d_\lambda \otimes \calQ^d_\mu$. We have inserted the factor of $\dim(\calP_\mu)$ for convenience: now trace-preservation of $\calC$ is equivalent to $\sum_{\mu} \tr_{\calQ^d_\mu} (J(\calC)_{\lambda \mu}) = I_{\calQ^d_\lambda}$ by \cite[Theorem 2.36, Item 3]{Wat18}. Hence we obtain
\begin{align*}
    \eqref{eq:trace_choi_projector_integral} & = \sum_{\substack{\lambda \vdash n \\ \ell(\lambda) \leq r}} \sum_{\substack{\mu \vdash m \\ \ell(\mu) \leq r}} \frac{\dim(\calP_\lambda)}{r^n} \cdot \tr \Big( J(\calC)_{\lambda \mu} \cdot \int_P q_\lambda(P)^T \otimes q_\mu(P) \cdot \dP\Big) \\
    & \leq \sum_{\substack{\lambda \vdash n \\ \ell(\lambda) \leq r}} \sum_{\substack{\mu \vdash m \\ \ell(\mu) \leq r}} \frac{\dim(\calP_\lambda)}{r^n} \cdot \tr \Big( J(\calC)_{\lambda \mu}\Big) \cdot \max_{\substack{\mu \vdash m \\ \ell(\mu) \leq r}} \norm{\int_P q_\lambda(P)^T \otimes q_\mu(P) \cdot \dP}_\infty \\
    & \leq \sum_{\substack{\lambda \vdash n \\ \ell(\lambda) \leq r}}  \frac{\dim(\calP_\lambda)}{r^n} \cdot \tr\big( I_{\calQ^d_\lambda}\big) \cdot \max_{\substack{\mu \vdash m \\ \ell(\mu) \leq r}} \norm{\int_P q_\lambda(P)^T \otimes q_\mu(P) \cdot \dP}_\infty \\
    & \leq  \sum_{\substack{\lambda \vdash n \\ \ell(\lambda) \leq r}} \frac{\dim(\calP_\lambda) \cdot \dim(\calQ^d_\lambda)}{r^n}  \cdot \max_{\substack{\mu \vdash m \\ \ell(\mu) \leq r}} \norm{\int_P q_\lambda(P)^T \otimes q_\mu(P) \cdot \dP}_\infty.
\end{align*}
This completes the proof.
\end{proof}

We can also reinterpret the bound of \Cref{lem:symmetric_channels_perfomance_upper_bounded} in terms of weak Schur sampling. Given a state of the form $P/r$, the probability of sampling $\blambda \sim \WSS_n(1^r/r)$ is $\dim(\calP_\lambda) \cdot \dim(\calQ^r_\lambda)/r^n$ for any $\lambda$ with $\ell(\lambda) \leq r$, so that the bound from this lemma becomes:
\begin{equation} \label{eq:bound_in_terms_of_WSS}
        \min_P \Big[ \tr\Big( \calC(\rho^{\otimes n}) \cdot P^{\otimes m} \Big) \Big] \leq \E_{\blambda} \bigg[ \frac{s_{\blambda}(1^d)}{s_{\blambda}(1^r)} \cdot \max_{ \substack{\mu \vdash m \\ \ell(\mu) \leq r}} \norm{ \int_P q_{\blambda}(P)^T \otimes q_\mu(P) \cdot \dP }_\infty \bigg].
\end{equation}

\paragraph{Step 4.} We now show that the choice of $\mu$ maximizing this upper bound, given $\lambda$, is $\mu = \lambda + k \cdot e_1$. Our main tool for this step is the following bound (in the PSD order) on the partially-transposed integral. We defer the proof to \Cref{sec:partially_transposed_integral_bound}, since it is relatively involved.

\begin{lemma}[The partially-transposed integral bound] \label{lem:transposed_integral}
    Let $\lambda \in \Par(n,d)$ and $\mu \in \Par(m,d)$, with $\ell(\lambda) \leq r$ and $\ell(\mu) \leq r$. Then we have
    \begin{equation*}
        \int q_{\lambda}^d(P)^T \otimes q_{\mu}^d(P) \cdot \dP \preceq \frac{s_{\lambda \cup \mu}(1^r)}{s_{\lambda \cup \mu}(1^d)} \cdot I_{\calQ^d_\lambda} \otimes I_{\calQ^d_\mu}. 
    \end{equation*}
\end{lemma}


We now analyze the products that appear in this upper bound. 

\begin{lemma} \label{lem:best_mu_is_lambda+ke1}
    Fix $\lambda \in \Par(n,d)$ with $\ell(\lambda) \leq r$. Then the ratio
    $s_{\lambda \cup \mu}(1^r)/s_{\lambda \cup \mu}(1^d)$
    is maximized over all $\mu \in \Par(n+k,d)$ by the choice $\mu = \lambda + k \cdot e_1$. For this $\mu$, we have
    \begin{equation*}
         \frac{s_{\lambda \cup \mu}(1^r)}{s_{\lambda \cup \mu}(1^d)} = \frac{s_{\lambda}(1^r)}{s_{\lambda }(1^d)} \cdot \prod_{j=1}^{k} \frac{r + \lambda_1 + j - 1}{d + \lambda_1 + j - 1}.
    \end{equation*}
\end{lemma}

\begin{proof}
    For fixed $\lambda$, consider the set of products of the form
    \begin{equation*}
        \frac{s_\lambda(1^d)}{s_\lambda(1^r)} \cdot \frac{s_{\lambda \cup \mu}(1^r)}{s_{\lambda \cup \mu}(1^d)}
    \end{equation*}
    for some $\mu \vdash n+k$. The denominator is nonzero since $\ell(\lambda) \leq r$. We can also assume $\ell(\mu) \leq r$, since otherwise $s_{\lambda \cup \mu}(1^r) = 0$. We first show that this product is maximized when $\mu = \lambda + k \cdot e_1$. Using the hook-content formula, Eq. \eqref{eq:hook_content_formula}, we have
    \begin{equation*}
        \frac{s_\lambda(1^d)}{s_\lambda(1^r)} \cdot \frac{s_{\lambda \cup \mu}(1^r)}{s_{\lambda \cup \mu}(1^d)} = \Big( \prod_{\Box \in \lambda} \frac{d + c(\Box)}{r + c(\Box)} \Big) \cdot \Big( \prod_{\Box \in \lambda \cup \mu} \frac{r + c(\Box)}{d + c(\Box)} \Big) = \prod_{\Box \in \mu \setminus \lambda} \frac{r + c(\Box)}{d+ c(\Box)}. 
    \end{equation*}
    We note that (i) each term in the product is positive, since $c(\Box) > -r$ for any $\Box \in \mu \setminus \lambda$, (ii) each term is at most $1$, and (iii) each term is an increasing function of the content. 
    Now consider any ordering of the boxes in $\mu \setminus \lambda$, such that after adding the first $j$ boxes to $\lambda$, we still have a valid Young diagram, for all $j$. The box added in the $j$-th step, $\Box_j$, has content at most $\lambda_1+j-1$, obtained by inserting that box into the first row. For $j \in [k]$, we therefore bound the corresponding term as
    \begin{equation*}
        \frac{r+c(\Box_j)}{d+c(\Box_j)} \leq \frac{r + (\lambda_1+j-1)}{d + (\lambda_1+j-1)}.
    \end{equation*}
    For the remaining boxes, we bound this term trivially by $1$. Thus,
    \begin{equation*}
        \prod_{\Box \in \mu \setminus \lambda} \frac{r + c(\Box)}{d+ c(\Box)} \leq \prod_{j=1}^{k} \frac{r + (\lambda_1 + j - 1)}{d + (\lambda_1 + j - 1)}.
    \end{equation*}
    This is exactly the product attained by setting $\mu = \lambda + k \cdot e_1$, as in this case $|\mu \setminus \lambda| = k$, and the $k$ boxes are in the first row. This completes the proof. \qedhere
\end{proof}


\paragraph{Step 5.} Finally, we prove the proposition. From the previous two steps, we have that
\begin{equation*}
    \min_{P} \Big[ \tr\Big( \calC(\rho^{\otimes n}) \cdot P^{\otimes m} \Big) \Big] \leq \E_{\blambda} \bigg[ \prod_{j=1}^{k} \frac{r + \blambda_1 + j - 1}{d + \blambda_1 + j - 1} \bigg].
\end{equation*}
If we have a rank-$r$ projector cloning channel $\calC$ which clones to global fidelity at least $1-\eps$ for any $P$, then we moreover have
\begin{equation*}
    1 - \epsilon \leq \min_{P} \Big[ \tr\Big( \calC(\rho^{\otimes n}) \cdot P^{\otimes m} \Big) \Big] \leq \E_{\blambda} \bigg[ \prod_{j=1}^{k} \frac{r + \blambda_1 + j - 1}{d + \blambda_1 + j - 1} \bigg]
\end{equation*}
We will now show that this implies the desired lower bound on $n$, for certain ranges of the parameters $d, r, \eps$. 


\begin{lemma} \label{lem:lower_bound_from_WSS}
    Suppose $\blambda \sim \WSS_n(1^r/r)$, and we have the inequality
    \begin{equation*}
    1 - \eps \leq \E_{\blambda} \bigg[ \prod_{j=1}^{k}\frac{r+\blambda_1+j-1}{d + \blambda_1+j-1}\bigg] 
    \end{equation*}
    Then if $d \geq 2$, $r \leq d/2$ and $\epsilon \leq 1/16$, we must have 
    \begin{equation*}
        n \geq \frac{1}{8} \cdot \frac{krd}{\eps}.
    \end{equation*}
\end{lemma}

\begin{proof}
We start by rewriting
\begin{equation*}
    \E_{\blambda} \bigg[ \prod_{j=1}^{k}\frac{r+\blambda_1+j-1}{d + \blambda_1+j-1}\bigg] =  \E \bigg[ \prod_{j=1}^{k} \Big( 1 - \frac{d-r}{d + \blambda_1 + j-1} \Big)\bigg].
\end{equation*}
We will now use the chain of inequalities: for $u_j \in [0,1]$,
\begin{equation*}
    \prod_{j=1}^k ( 1 - u_j ) \leq \prod_{j=1}^{k} \frac{1}{1+u_j} \leq \frac{1}{1 + \sum_{j=1}^k u_j} \leq \frac{1}{1 + k \cdot \min(u_j)}.
\end{equation*}
In particular, we will set $u_j \leftarrow (d-r)/(d + \blambda_1 + j - 1)$, noting that $\blambda_1 + j - 1 \geq 0$ so that $u_j \leq (d-r)/d < 1$ as we require. Note also that $\min u_j = (d-r)/(d+\blambda_1 + k - 1)$. Then we get
\begin{align*}
    \E \bigg[ \prod_{j=1}^{k} \Big( 1 - \frac{d-r}{d + \blambda_1 + j-1} \Big)\bigg] \leq \E \bigg[ \frac{1}{1 + \frac{k(d-r) }{ ( d + \blambda_1 + k-1)}}\bigg] & = 1 - \E \bigg[ \frac{k(d-r)}{d + \blambda_1 + k - 1 + k(d-r)}\bigg] \\
    & \leq 1 - \frac{k(d-r)}{d + \E[\blambda_1] + k-1 + k(d-r)}.
\end{align*}
The last inequality is Jensen's. We now use 
\begin{equation*}
    \E [ \blambda_1] \leq \frac{n}{r} + 2 \sqrt{n} \leq \frac{2n}{r} + r.
\end{equation*}
The first inequality is \cite[Theorem 5.2]{OW17a}, and the second is AM-GM, applied as $2\sqrt{n} = 2\sqrt{ r \cdot (n/r)}$. Then we have
\begin{equation*}
    1 - \epsilon \leq 1 - \frac{k(d-r)}{d + \E[\blambda_1] + k-1 + k(d-r)} \leq 1 - \frac{k(d-r)}{d + \frac{2n}{r}  + r + k-1 + k(d-r)}.
\end{equation*}
We can rearrange this inequality to find a lower bound on $n$:
\begin{equation*}
    n \geq \frac{1}{2} \cdot \bigg(\frac{kr(d-r)}{\eps} - r\Big(d + r + k -1 + k(d-r) \Big) \bigg) \geq \frac{1}{2} \cdot \bigg( \frac{kr (d-r)}{\eps} - 4krd\bigg).
\end{equation*}
Since $r \geq 1$, we need $d \geq 2$ for this bound to be nontrivial. Then, if $r \leq d/2$, we have $kr(d-r)/\eps \geq krd/2\eps$. Finally, because $\eps \leq 1/16$, we have $4krd \leq krd/4\eps$. For these parameter restrictions, we then get
\begin{equation*}
    n \geq \frac{1}{8} \cdot \frac{krd}{\eps}. \qedhere
\end{equation*}
\end{proof}
This lemma completes the proof of \Cref{prop:projector_lower_bound}.

\subsection{Sample-optimal cloning of mixed states}

Here, we combine our ingredients to conclude our main result, which we restate here for the reader's convenience. 


\begin{theorem}[The PCT cloner is sample-optimal; \Cref{thm:main_result_intro}, restated] 
    Suppose $\calC: \states(\C^d)^{\otimes n} \to \states(\C^d)^{\otimes (n+k)}$ is a cloning channel such that
    \begin{equation*}
         \Fid \big(  \calC(\rho^{\otimes n}) , \rho^{\otimes (n+k)}\big) \geq 1-\eps,
    \end{equation*}
    for all states $\rho$ of rank at most $r$. Then we must have $n = \Omega(krd/\eps)$, for $d \geq 2$, and $\epsilon \leq \epsilon_0$, where $\epsilon_0$ is a universal constant. Moreover, the PCT cloner achieves this guarantee with $n = O(krd/\eps)$ copies. Thus, the sample complexity of cloning $k$ additional copies of any unknown rank-$r$ input state $\rho^{\otimes n}$ to fidelity $1-\eps$ is $\Theta(krd/\eps)$. 
\end{theorem}

\begin{proof}
    In \Cref{prop:upper_bound_PCT}, we show the PCT cloner attains the upper bound. In \Cref{prop:projector_lower_bound}, we show that even when promised the input is a rank-$r$ projector state, we require $n = \Omega(krd/\eps)$ copies for $d \geq 2$, $r \leq d/2$, and $\epsilon \leq 1/16$. The lower bound for generic rank follows by considering projector states of rank $\min(r, d/2)$. This completes the proof.
\end{proof}

\subsection{The partially-transposed integral} \label{sec:partially_transposed_integral_bound}

 Throughout this section, it will be useful to introduce some shorthand to simplify our expressions. In particular, we will adopt the convention of writing $X_\lambda$ as shorthand for $q_\lambda^d(X)$. The dimension $d$ can be inferred from $\lambda$, since we will always specify the number of parts of $\lambda$, e.g.\ as $\lambda \in \Par(-,d)$. As an example, we will write $I_\lambda$ as shorthand for $q^d_\lambda(I) = I_{\calQ^d_\lambda}$. 

We now circle back to the partially-transposed integral, and prove the bound we claimed earlier, which we restate for convenience.

\begin{lemma}[The partially-transposed integral bound; \Cref{lem:transposed_integral} restated] \label{lem:transposed_integral_2}
    Let $\lambda \in \Par(-,d)$ and $\mu \in \Par(-,d)$. Then
    \begin{equation*}
        \int P^T_\lambda \otimes P_\mu \cdot \dP   \preceq \frac{s_{\lambda \cup \mu}(1^r)}{s_{\lambda \cup \mu}(1^d)} \cdot I_{\lambda} \otimes I_{\mu} = \Phi_{\lambda \cup \mu}^{(d,r)} \cdot I_\lambda \otimes I_\mu. 
    \end{equation*}
\end{lemma}

We take a somewhat indirect approach. We start by considering the partial transpose of this integral, resulting in both terms in the integrand un-transposed. The un-transposed integral bound reads
\begin{equation} \label{eq:transposed_bound}
        \int P_\lambda \otimes P_\mu \cdot \dP \preceq \frac{s_{\lambda \cup \mu}(1^r)}{s_{\lambda \cup \mu}(1^d)} \cdot I_{\lambda} \otimes I_{\mu}.
\end{equation}
This bound itself is actually relatively easy to show, and follows from the Littlewood--Richardson rule. However, even with \Cref{eq:transposed_bound} we cannot conclude \Cref{lem:transposed_integral}, as it is not necessarily true that $X \preceq Y$ implies $X^\Gamma \preceq Y^\Gamma$.\footnote{For instance, take $X = 0$ and $Y = \ketbra{\Omega}$, with $\ket{\Omega} = \ket{00} + \ket{11}$. Then $X^\Gamma = 0$ and $Y^\Gamma = \mathrm{SWAP}$, which has both positive and negative eigenvectors.} We overcome this by showing a stronger structural result: the difference between the right- and left-hand sides is not only PSD, but separable. Defining
\begin{equation*}
    \Delta_{\lambda \mu} \coloneq \frac{s_{\lambda \cup \mu}(1^r)}{s_{\lambda \cup \mu}(1^d)} \cdot I_{\lambda} \otimes I_{\mu} - \int P_\lambda \otimes P_\mu \cdot \dP, 
\end{equation*}
we will show the following result.

\begin{proposition} \label{prop:untransposed_bound}
    Let $\lambda \in \Par(-,d)$ and $\mu \in \Par(-,d)$. Then we have $$\Delta_{\lambda \mu} \in \Sep_+(\calQ_{\lambda}^d : \calQ_{\mu}^d).$$ 
    That is, $\Delta_{\lambda \mu}$ is a nonnegative linear combination of matrices of the form $A \otimes B$, where $A$ acts on $\calQ^d_\lambda$, $B$ acts on $\calQ^d_\mu$, and $A, B \succeq 0$. 
\end{proposition}

From here, we can prove the partially-transposed integral bound straightforwardly.

\begin{proof}[Proof of \Cref{lem:transposed_integral}]
    Since $A^T$ is PSD if $A$ is, we have
    \begin{equation*}
        \Sep_+(\calQ_{\lambda}^d : \calQ_{\mu}^d)^\Gamma \subseteq \Sep_+(\calQ_{\lambda}^d : \calQ_{\mu}^d).
    \end{equation*}
    Thus, $\Delta_{\lambda \mu}^\Gamma \in \Sep_+(\calQ_{\lambda}^d : \calQ_{\mu}^d)$, and in particular,
    \begin{equation*}
        \Delta_{\lambda \mu}^\Gamma = \frac{s_{\lambda \cup \mu}(1^r)}{s_{\lambda \cup \mu}(1^d)} \cdot I_\lambda \otimes I_\mu - \int P_\lambda^T \otimes P_\mu \cdot \dP \succeq 0.
    \end{equation*}
    This completes the proof.
\end{proof}

The rest of the section is devoted to establishing the proposition. We will start by developing notation, defining useful quantities, and establishing some lemmas.

\subsubsection{Background for the proof}

We will write $S_{\lambda \mu}^{(d,r)}$ as shorthand for the integral appearing in \Cref{eq:transposed_bound} ($S$ is for \emph{sum}). We note that $S^{(d,r)}_{\lambda \mu}$ can be rewritten as an integral over all possible rotations of a \emph{fixed} matrix: for any fixed rank-$r$ projector $P^{(r)}_0 \in \C^{d \times d}$,
\begin{equation} \label{eq:S}
    S^{(d,r)}_{\lambda \mu} = \int_{U \in U(d)} \big( U_\lambda \otimes U_\mu \big) \cdot \big( P^{(r)}_{0,\lambda}  \otimes P^{(r)}_{0,\mu} \big) \cdot \big( U_\lambda \otimes U_\mu \big)^\dagger \cdot \dU, 
\end{equation}
It will be convenient to set $P^{(r)}_0 \leftarrow \sum_{i = d-r+1}^d \ketbra{i}$, i.e.\ $P^{(r)}_0$ is projection onto the \emph{last} $r$ standard basis vectors of $\C^d$.

\paragraph{A few useful projectors on $\calQ^d_\lambda$.} Let $\lambda \in \Par(-,d)$, and $\lambda' \in \Par(-,d)$ such that $\lambda'$ interlaces $\lambda$. We now define a few useful projectors on the space $\calQ^d_\lambda$. First, we define $\Pi_{\lambda' \subseteq \lambda}$ as the projector
\begin{equation*}
    \Pi_{\lambda' \subseteq \lambda} \coloneq \Proj\big\{ T \, : \,\sh(T_{\leq (d-1)}) = \lambda' \big\}. 
\end{equation*}
That is, $\Pi_{\lambda' \subseteq \lambda}$ projects only onto those SSYTs of shape $\lambda$ that, when you look only at the symbols less than $d$, has shape $\lambda'$. Note that $\sum_{\lambda'} \Pi_{\lambda' \subseteq \lambda} = I_\lambda$. Next, we define $\Pi^{(t)}_{\lambda}$ for $t \in [d]$ as the projector 
\begin{equation*}
    \Pi^{(t)}_{\lambda} \coloneq \Proj\big\{ T \, : \sh(T_{\leq(d-t)}) = \varnothing \big\}. 
\end{equation*}
In words: $\Pi^{(t)}_{\lambda}$ projects onto those SSYTs of shape $\lambda$ that are filled only by the \emph{last} $t$ symbols in the alphabet. Note that $\Pi^{(d)}_\lambda = I_\lambda$ and $\Pi^{(r)}_\lambda = (P_0^{(r)})_\lambda$. Finally, we define $\Pi^{(t)}_{\lambda' \subseteq \lambda}$ as the product of the two projectors above. That is, 
\begin{equation*}
    \Pi^{(t)}_{\lambda' \subseteq \lambda} \coloneq \Proj\big\{ T \, : \,\sh(T_{\leq (d-1)}) = \lambda' \; \mathrm{and} \; \sh(T_{\leq(d-t)}) = \varnothing \big\}. 
\end{equation*}

\paragraph{Embedding maps.} We also define embedding maps $E_{\lambda' \subseteq \lambda}$ that include $\calQ^{d-1}_{\lambda'} \hookrightarrow \calQ^d_\lambda$. Concretely
$E_{\lambda' \subseteq \lambda} \cdot \ket{T'} = \ket{T}$,
where $T$ is the SSYT of $\lambda$ such that $T_{\leq (d-1)} = T'$. We note the following relations which we will use later. First, for any $V \in U(d-1)$, viewed as an element of $U(d)$ by letting it act on the first $d-1$ basis vectors as $V$, and fixing $\ket{d}$, we have 
\begin{equation}\label{eq:aux_1}
    E_{\lambda' \subseteq \lambda} \cdot V_{\lambda'} = V_\lambda \cdot E_{\lambda' \subseteq \lambda}.
\end{equation}
Second, we have
\begin{equation}\label{eq:aux_2}
    E_{\lambda' \subseteq \lambda} \cdot I_{\lambda'} \cdot E_{\lambda' \subseteq \lambda}^\dagger = \Pi_{\lambda' \subseteq \lambda}. 
\end{equation}
Finally, 
\begin{equation}\label{eq:aux_3}
    E_{\lambda' \subseteq \lambda} \cdot \Pi^{(t-1)}_{\lambda'} \cdot E_{\lambda' \subseteq \lambda}^\dagger = \Pi^{(t)}_{\lambda' \subseteq \lambda}.
\end{equation}
To understand this last one, note that on the left-hand side, we are embedding all SSYTs of shape $\lambda'$ with alphabet $\{(d-1)-(t-1)+1, \dots, d-1\} = \{d-t+1, \dots, d-1\}$ into $\calQ^d_\lambda$, which gives us the right-hand side. We also note $\Pi^{(t)}_{\lambda} = \sum_{\lambda' \preceq \lambda} \Pi^{(t)}_{\lambda' \subseteq \lambda}$. 

For simpler notation, we will also introduce an \emph{embedding channel}, defined as $\calE_{\lambda' \subseteq \lambda}(X) \coloneq E_{\lambda' \subseteq \lambda} \cdot X \cdot E_{\lambda' \subseteq \lambda}^\dagger$. 

\paragraph{Lifting.} Finally, let $\mu \in \Par(-,d)$ as well, and let $\mu' \in \Par(-,d-1)$ such that $\mu' \preceq \mu$. Given any operator $X$ acting on $\calQ^{d-1}_{\lambda'} \otimes \calQ^{d-1}_{\mu'}$, we can use the embedding maps $E_{\lambda' \subseteq \lambda}$ and $E_{\mu' \subseteq \mu}$ to lift $X$ to an operator on $\calQ^d_\lambda \otimes \calQ^d_\mu$. In our case, it will be convenient to average the resulting operator over all possible rotations as well, so as to bear closer resemblance to $S_{\lambda\mu}$. Thus, we define the $\emph{lift}$ as the map
\begin{equation*}
    \Lift_{\lambda'\mu'}^{\lambda \mu}(X) \coloneq \int_{U \in U(d)} \big(U_{\lambda} \otimes U_\mu\big) \cdot \big(E_{\lambda' \subseteq \lambda} \otimes E_{\mu' \subseteq \mu} \big) \cdot X  \cdot \big(E_{\lambda' \subseteq \lambda} \otimes E_{\mu' \subseteq \mu} \big)^\dagger \cdot \big(U_{\lambda} \otimes U_\mu\big)^\dagger \cdot \dU.
\end{equation*}
A key observation is that if $X \in \Sep^+(\calQ^{d-1}_{\lambda'} : \calQ^{d-1}_{\mu'})$, then $\Lift_{\lambda'\mu'}^{\lambda \mu}(X) \in \Sep^+(\calQ^{d}_\lambda : \calQ^d_\mu)$. 

We can simplify our notation if we further introduce a \emph{twirling-by-unitaries channel}, given by \begin{equation*}
\calU_{\lambda \mu}(X) \coloneq \int_{U \in U(d)} (U_\lambda \otimes U_\mu) \cdot X \cdot (U_\lambda \otimes U_\mu)^\dagger \cdot \dU. 
\end{equation*}
Then 
\begin{equation*}
    \Lift_{\lambda'\mu'}^{\lambda \mu}(X) = \calU_{\lambda \mu} ( \calE_{\lambda' \subset \lambda} \otimes \calE_{\mu' \subset \mu} (X)). 
\end{equation*}

Our first key lemma gives a recursive formula for $I_\lambda \otimes I_\mu$, in terms of the lifts of operators of the form $I_{\lambda'} \otimes I_{\mu'}$, for intertwining diagrams $\lambda' \preceq \lambda$ and $\mu' \preceq \mu$.

\begin{lemma} \label{lem:lift_recursion_identity}
    Fix $\lambda, \mu \in \Par(-,d)$, and $\lambda' \in \Par(-,d-1)$ with $\lambda' \preceq \lambda$. Then
        \begin{equation*}
            \sum_{\mu' \preceq \mu} \Lift^{\lambda \mu}_{\lambda'\mu'} \big( I_{\lambda'} \otimes I_{\mu'} \big) = \frac{s_{\lambda'}(1^{d-1})}{s_{\lambda}(1^d)} \cdot I_{\lambda} \otimes I_{\mu}.
        \end{equation*}
\end{lemma}

\begin{proof}
    From \Cref{eq:aux_2}, we have
    \begin{equation*}
        \big(E_{\lambda' \subseteq \lambda} \otimes E_{\mu' \subseteq \mu} \big) \cdot  \big( I_{\lambda'} \otimes I_{\mu'} \big) \cdot \big(E_{\lambda' \subseteq \lambda} \otimes E_{\mu' \subseteq \mu} \big)^\dagger = \Pi_{\lambda' \subseteq \lambda} \otimes \Pi_{\mu' \subseteq \mu}.
    \end{equation*}
    Thus, 
    \begin{align*}
        \sum_{\mu' \preceq \mu} \Lift^{\lambda \mu}_{\lambda'\mu'} \big( I_{\lambda'} \otimes I_{\mu'} \big) &  = \sum_{\mu' \preceq \mu} \int \big(U_{\lambda} \otimes U_\mu\big) \cdot \big( \Pi_{\lambda' \subseteq \lambda} \otimes \Pi_{\mu' \subseteq \mu}\big) \cdot \big(U_{\lambda} \otimes U_\mu\big)^\dagger \cdot \dU  \\
        & = \int \big(U_{\lambda} \otimes U_\mu\big) \cdot \big( \Pi_{\lambda' \subseteq \lambda} \otimes I_\mu \big) \cdot \big(U_{\lambda} \otimes U_\mu\big)^\dagger \cdot \dU \\
        & = \bigg( \int U_\lambda \cdot \Pi_{\lambda' \subseteq \lambda} \cdot U_\lambda^\dagger \cdot \dU \bigg) \otimes I_{\mu} \\
        & = \frac{s_{\lambda'}(1^{d-1})}{s_{\lambda}(1^d)} \cdot I_{\lambda} \otimes I_{\mu}. 
    \end{align*}
    The second step uses $\sum_{\mu' \preceq \mu} \Pi_{\mu' \subseteq \mu} = I_\mu$. The last step is Schur's lemma, where we have used $\tr(\Pi_{\lambda' \subseteq \lambda}) = \dim(\calQ^{d-1}_{\lambda'}) = s_{\lambda'}(1^{d-1})$. 
\end{proof}

Our second key lemma gives a similar recursive formula, except for the integral $S^{(d,r)}_{\lambda \mu}$. 

\begin{lemma} \label{lem:lift_recursion_integral}
    Fix $\lambda, \mu \in \Par(-,d)$, and $\lambda' \in \Par(-,d-1)$ such that $\lambda' \preceq \lambda$ and $\ell(\lambda') \leq r-1$. Then 
    \begin{equation*}
        S^{(d,r)}_{\lambda \mu} =  \frac{s_{\lambda}(1^{r})}{s_{\lambda'}(1^{r-1})} \sum_{\mu' \preceq \mu} \Lift^{\lambda \mu}_{\lambda'\mu'} \Big( S^{(d-1,r-1)}_{\lambda'\mu'} \Big).
    \end{equation*}
\end{lemma}

\newcommand{\dV}{\mathrm{d}V}
\newcommand{\dW}{\mathrm{d}W}

\begin{proof}

    From \Cref{eq:S} and the fact that $P^{(r)}_{0,\lambda} = \Pi^{(r)}_{\lambda}$, we have
\begin{equation*}
    S_{\lambda \mu}^{(d,r)} = \int_{U \in U(d)} \big( U_\lambda \otimes U_\mu \big) \cdot \big( \Pi^{(r)}_\lambda  \otimes \Pi^{(r)}_{\mu} \big) \cdot \big( U_\lambda \otimes U_\mu \big)^\dagger \cdot \dU.
\end{equation*}
    Embed $U(r) \hookrightarrow U(d)$ by letting $W \in U(r)$ act nontrivially on the \emph{last} $r$ basis elements. Note then that
    \begin{equation} \label{eq:aux_4}
        \int_{W \in U(r)} W_\lambda \cdot \Pi_{\lambda' \subseteq \lambda}^{(r)} \cdot W_\lambda^\dagger \cdot \dW = \frac{\tr\big( \Pi^{(r)}_{\lambda' \subseteq \lambda} \big)}{s_\lambda(1^r)} \cdot \Pi^{(r)}_{\lambda} = \frac{s_{\lambda'}(1^{r-1})}{s_\lambda(1^r)} \cdot \Pi^{(r)}_{\lambda}. 
    \end{equation}
    In the above, we have used that $\Pi_{\lambda' \subseteq \lambda}^{(r)}$ is a sub-projector of $\Pi_{\lambda}^{(r)}$, and Schur's lemma, noting that the integral is $U(r)$-invariant on the support of $\Pi^{(r)}_\lambda$, a subspace isomorphic to $\calQ^r_\lambda$. Thus
    \begin{align}
         \frac{s_{\lambda'}(1^{r-1})}{s_\lambda(1^r)} \cdot S^{(d,r)}_{\lambda \mu} & = \frac{s_{\lambda'}(1^{r-1})}{s_\lambda(1^r)} \cdot  \int_{U \in U(d)} \big(U_{\lambda} \otimes U_{\mu}\big)  \cdot \Big( \Pi^{(r)}_\lambda \otimes \Pi^{(r)}_\mu \Big) \cdot \big(U_{\lambda} \otimes U_{\mu}\big)^\dagger \cdot \dU \nonumber \\
         & = \int_{U \in U(d)} \int_{W \in U(r)} \big((UW)_{\lambda} \otimes U_{\mu}\big)  \cdot \Big( \Pi^{(r)}_{\lambda' \subseteq \lambda} \otimes \Pi^{(r)}_\mu \Big) \cdot \big((UW)_{\lambda} \otimes U_{\mu}\big)^\dagger \cdot \dW \cdot \dU \nonumber \\
         & = \int_{U \in U(d)} \int_{W \in U(r)} \big((UW)_{\lambda} \otimes (UW)_{\mu}\big)  \cdot \Big( \Pi^{(r)}_{\lambda' \subseteq \lambda} \otimes\Pi^{(r)}_\mu \Big) \cdot \big((UW)_{\lambda} \otimes (UW)_{\mu}\big)^\dagger \cdot \dW \cdot \dU \nonumber \\
         & = \int_{U \in U(d)} \big(U_{\lambda} \otimes U_{\mu}\big)  \cdot \Big( \Pi^{(r)}_{\lambda' \subseteq \lambda} \otimes \Pi^{(r)}_\mu \Big) \cdot \big(U_{\lambda} \otimes U_{\mu}\big)^\dagger \cdot \dU. \label{eq:halfway_through}
    \end{align}
    In the second step, we have used \Cref{eq:aux_4}. In the third step, we have used that $P_0^{(r)}$ is invariant under $W \in U(r)$, and hence $(P_0^{(r)})_\mu = \Pi^{(r)}_\mu$ is invariant under $W_\mu$. In the fourth step, we have used the right-invariance of the Haar measure, and redefined $UW \to U$. Finally, we use $\Pi^{(r)}_\mu = \sum_{\mu' \preceq \mu} \Pi^{(r)}_{\mu' \subseteq \mu}$, and \Cref{eq:aux_3} to get:
    \begin{align}
        \eqref{eq:halfway_through} 
        & = \sum_{\mu' \preceq \mu} \int_{U \in U(d)} \big(U_{\lambda} \otimes U_{\mu}\big)  \cdot \Big( \Pi^{(r)}_{\lambda' \subseteq \lambda} \otimes \Pi^{(r)}_{\mu' \subseteq \mu} \Big) \cdot \big(U_{\lambda} \otimes U_{\mu}\big)^\dagger \cdot \dU \nonumber \\
        & = \sum_{\mu' \preceq \mu} \int_{U \in U(d)} \big(U_{\lambda} \otimes U_{\mu}\big)  \cdot \big( E_{\lambda' \subseteq \lambda} \otimes E_{\mu' \subseteq \mu}\big) \cdot \Big( \Pi^{(r-1)}_{\lambda'} \otimes \Pi^{(r-1)}_{\mu'} \Big) \cdot \big( E_{\lambda' \subseteq \lambda} \otimes E_{\mu' \subseteq \mu}\big)^\dagger \cdot \big(U_{\lambda} \otimes U_{\mu}\big)^\dagger \cdot \dU. \label{eq:three-quarters-there}
    \end{align}
    Finally, we use the right-invariance of the Haar measure again, this time to re-insert a unitary $V \in U(d-1)$ as $U \to U \cdot V$ (viewing the action of $V$ on $U(d)$ as acting nontrivially only on the first $d-1$ basis elements), which we can then pass through the embedding maps using \Cref{eq:aux_1}. This gives us 
    \begin{align*}
        \eqref{eq:three-quarters-there} 
        & =\sum_{\mu' \preceq \mu} \int_{U \in U(d)} \int_{V \in U(d-1)} \big((UV)_{\lambda} \otimes (UV)_{\mu}\big)  \cdot \big( E_{\lambda' \subseteq \lambda} \otimes E_{\mu' \subseteq \mu}\big) \cdot \Big( \Pi^{(r-1)}_{\lambda'} \otimes \Pi^{(r-1)}_{\mu'} \Big) \cdot \\
        & \qquad \qquad \qquad \qquad \qquad \qquad  \qquad \qquad \quad \; \qquad \qquad \qquad  \qquad \cdot \big( E_{\lambda' \subseteq \lambda} \otimes E_{\mu' \subseteq \mu}\big)^\dagger \cdot \big((UV)_{\lambda} \otimes (UV)_{\mu}\big)^\dagger \cdot \dV \cdot \dU \\
        & = \sum_{\mu' \preceq \mu} \int_{U \in U(d)}  \big(U_{\lambda} \otimes U_{\mu}\big)  \cdot \big( E_{\lambda' \subseteq \lambda} \otimes E_{\mu' \subseteq \mu}\big) \cdot \\
        & \qquad \cdot \bigg( \int_{V \in U(d-1)} \big(V_{\lambda'} \otimes V_{\mu'}\big) \cdot \Big( \Pi^{(r-1)}_{\lambda'} \otimes \Pi^{(r-1)}_{\mu'} \Big) \cdot \big(V_{\lambda'} \otimes V_{\mu'}\big)^\dagger \cdot \dV \bigg) \cdot \big( E_{\lambda' \subseteq \lambda} \otimes E_{\mu' \subseteq \mu}\big)^\dagger \cdot \big(U_{\lambda} \otimes U_{\mu}\big)^\dagger \cdot \dU \\
        & = \sum_{\mu' \preceq \mu} \int_{U \in U(d)}  \big(U_{\lambda} \otimes U_{\mu}\big)  \cdot \big( E_{\lambda' \subseteq \lambda} \otimes E_{\mu' \subseteq \mu}\big)  \cdot S^{(d-1,r-1)}_{\lambda'\mu'}  \cdot \big( E_{\lambda' \subseteq \lambda} \otimes E_{\mu' \subseteq \mu}\big)^\dagger \cdot \big(U_{\lambda} \otimes U_{\mu}\big)^\dagger \cdot \dU \\
        & = \sum_{\mu' \preceq \mu} \Lift_{\lambda' \mu'}^{\lambda \mu}\Big(S^{(d-1,r-1)}_{\lambda'\mu'}\Big).
    \end{align*}
    This completes the proof. 
\end{proof}


\paragraph{A Schur polynomial identity.} We will need one last technical lemma. For a diagram $\lambda \in \Par( - , d)$, write $\lambda^-$ for the diagram in $\Par(-, d-1)$ with $\lambda^- = (\lambda_2, \lambda_3, \dots, \lambda_d)$. 

\begin{lemma} \label{lem:lambda_minus_identity}
    Suppose $\lambda, \mu \in \Par(-, d)$, with $\lambda_1 \geq \mu_1$, and such that $\ell(\lambda) \leq r$ and $\ell(\mu) \leq r$. Then 
    \begin{equation*}
        \frac{s_{\lambda}(1^r)}{s_{\lambda}(1^d)} \cdot \frac{s_{\lambda^-}(1^{d-1})}{s_{\lambda^-}(1^{r-1})} = \frac{s_{\lambda \cup \mu}(1^r)}{s_{\lambda \cup \mu}(1^d)} \cdot \frac{s_{\lambda^-\cup \mu^-}(1^{d-1})}{s_{\lambda^- \cup \mu^-}(1^{r-1})}. 
    \end{equation*}
    Rearranging this equation gives:
    \begin{equation*}
        \frac{s_{\lambda \cup \mu}(1^r)}{s_{\lambda \cup \mu}(1^d)} \cdot \frac{s_{\lambda}(1^{d})}{s_{\lambda^-}(1^{d-1})} = \frac{s_{\lambda^-\cup \mu^-}(1^{r-1})}{s_{\lambda^- \cup \mu^-}(1^{d-1})} \cdot \frac{s_{\lambda}(1^r)}{s_{\lambda^-}(1^{r-1})}.
    \end{equation*}
\end{lemma}

\begin{proof}
    By the hook-content formula, \Cref{eq:hook_content_formula}, we have
    \begin{equation*}
        \frac{s_{\lambda}(1^r)}{s_{\lambda}(1^d)} = \prod_{\Box \in \lambda} \frac{r + c(\Box)}{d + c(\Box)} = \bigg(\prod_{\Box \in \mathrm{Row}_1(\lambda)} \frac{r + c(\Box)}{d + c(\Box)}\bigg) \cdot \bigg(\prod_{\Box \notin \mathrm{Row}_1(\lambda)} \frac{r + c(\Box)}{d + c(\Box)}\bigg).
    \end{equation*}
    Here, we have split the product up by whether the box is in the first row of $\lambda$, or not. Now, we also have
    \begin{equation*}
        \frac{s_{\lambda^-}(1^{r-1})}{s_{\lambda^-}(1^{d-1})} = \prod_{\Box \in \lambda^-} \frac{(r-1) + c(\Box)}{(d-1) + c(\Box)} = \prod_{\Box \notin \mathrm{Row}_1(\lambda)} \frac{(r-1) + (c(\Box)+1)}{(d-1) + (c(\Box)+1)} = \prod_{\Box \notin \mathrm{Row}_1(\lambda)} \frac{r + c(\Box)}{d + c(\Box)}. 
    \end{equation*}
    The second equation holds because the boxes of $\lambda^-$ are in one-to-one correspondence with the boxes of $\lambda$ which are not in the first row -- simply shift those boxes of $\lambda$ up once. This shift increases the content of each box by one. Combining these two equations gives us
    \begin{equation*}
         \frac{s_{\lambda}(1^r)}{s_{\lambda}(1^d)} \cdot \frac{s_{\lambda^-}(1^{d-1})}{s_{\lambda^-}(1^{r-1})} = \prod_{\Box \in \mathrm{Row}_1(\lambda)} \frac{r + c(\Box)}{d + c(\Box)}. 
    \end{equation*}
    This ratio is well-defined since $\ell(\lambda) \leq r$, so that $\ell(\lambda^-) \leq r-1$. 
    Moreover, this equation holds not only for $\lambda$, but for any diagram with length at most $r$. In particular, $\lambda \cup \mu$ has this property, and so we also have 
    \begin{equation*}
         \frac{s_{\lambda \cup \mu}(1^r)}{s_{\lambda \cup \mu}(1^d)} \cdot \frac{s_{(\lambda\cup \mu)^-}(1^{d-1})}{s_{(\lambda \cup \mu)^-}(1^{r-1})} = \prod_{\Box \in \mathrm{Row}_1(\lambda \cup \mu)} \frac{r + c(\Box)}{d + c(\Box)}. 
    \end{equation*}
    The lemma then follows from two observations: (i) $\lambda$ and $\lambda \cup \mu$ have the same boxes in their first row, since $\lambda_1 \geq \mu_1$; and (ii) $(\lambda \cup \mu)^- = \lambda^- \cup \mu^-$. 
\end{proof}

\subsubsection{Proving \Cref{prop:untransposed_bound}}

We are now ready to prove the proposition, which we restate here for convenience.

\begin{proposition}[\Cref{prop:untransposed_bound}, restated]
    Let $\lambda \in \Par(-,d)$ and $\mu \in \Par(-,d)$. Then we have $$\Delta_{\lambda \mu} \in \Sep_+(\calQ_{\lambda}^d : \calQ_{\mu}^d).$$ 
\end{proposition}

\begin{proof}
    First, we can assume $\ell(\lambda) \leq r$ and $\ell(\mu) \leq r$, since otherwise $\Delta_{\lambda \mu} = 0$. For example, if $\ell(\lambda) > r$, then $P_\lambda = 0$, which makes the integral vanish, and also $\ell(\lambda \cup \mu) \geq \ell(\lambda) > r$, which makes $s_{\lambda \cup \mu}(1^r) = 0$. Thus $\Delta_{\lambda \mu} = 0$, and the case where $\ell(\mu) > r$ is similar.
    
    We proceed by induction. Our base case will be $r=1$ and $d \geq r$ arbitrary. Assume $\mu_1 \leq \lambda_1$. Then since $\lambda$ and $\mu$ are just one-row diagrams, $\mu \subseteq \lambda$, and moreover, $\lambda \cup \mu = \lambda$. We then have
    \begin{align*}
        \Delta_{\lambda \mu} & = \Big( \frac{s_{\lambda}(1^r)}{s_{\lambda}(1^d)} \cdot I_{\lambda} \otimes I_{\mu}\Big) - \Big( \int q_\lambda^d(P) \otimes q^d_\mu(P) \cdot \dP \Big) \\
        & = \Big( \int q_\lambda^d(P) \cdot \dP \Big) \otimes I_\mu - \Big( \int q_\lambda^d(P) \otimes q^d_\mu(P) \cdot \dP \Big) \\
        & = \int q_\lambda^d(P) \otimes \big( I_\mu - q^d_\mu(P) \big) \cdot \dP. 
    \end{align*}
    Since $q^d_\mu(P) \preceq I_\mu$, we have $\Delta_{\lambda \mu} \in \Sep_+( \calQ^d_\lambda : \calQ^d_\mu)$, as we wanted to show. The case where $\mu_1 \geq \lambda_1$ is symmetric. 

    Now assume we have shown the result for rank $r-1$ and arbitrary dimension at least $r-1$. We will prove the result for rank $r$ and dimension $d \geq r$. We also assume $\lambda_1 \geq \mu_1$ --- the proof is symmetric in the other case. We have
    \begin{align*}
        \Delta_{\lambda \mu} & =  \frac{s_{\lambda \cup \mu}(1^r)}{s_{\lambda \cup \mu}(1^d)} \cdot I_{\lambda} \otimes I_{\mu}  - S^{(d,r)}_{\lambda \mu} \\
        & = \bigg(\frac{s_{\lambda \cup \mu}(1^r)}{s_{\lambda \cup \mu}(1^d)} \cdot \frac{s_{\lambda}(1^{d})}{s_{\lambda^-}(1^{d-1})} \cdot \sum_{\mu'\preceq \mu}\Lift^{\lambda \mu}_{\lambda^-\mu'} \Big( I_{\lambda^-} \otimes I_{\mu'} \Big) \bigg) - \bigg( \frac{s_{\lambda}(1^r)}{s_{\lambda^-}(1^{r-1})} \sum_{\mu' \preceq \mu} \Lift^{\lambda \mu}_{\lambda^-\mu'} \Big(  S^{(d-1,r-1)}_{\lambda^-\mu'}\Big) \bigg) \\
        & = \bigg(\frac{s_{\lambda^-\cup \mu^-}(1^{r-1})}{s_{\lambda^- \cup \mu^-}(1^{d-1})} \cdot \frac{s_{\lambda}(1^r)}{s_{\lambda^-}(1^{r-1})} \cdot \sum_{\mu'\preceq \mu}\Lift^{\lambda \mu}_{\lambda^-\mu'} \Big( I_{\lambda^-} \otimes I_{\mu'} \Big) \bigg) - \bigg( \frac{s_{\lambda}(1^r)}{s_{\lambda^-}(1^{r-1})} \sum_{\mu' \preceq \mu} \Lift^{\lambda \mu}_{\lambda^-\mu'} \Big(  S^{(d-1,r-1)}_{\lambda^-\mu'}\Big) \bigg) \\
        & = \frac{s_{\lambda}(1^r)}{s_{\lambda^-}(1^{r-1})} \cdot \sum_{\mu' \preceq \mu} \Lift^{\lambda \mu}_{\lambda^-\mu'} \bigg( \frac{s_{\lambda^-\cup \mu^-}(1^{r-1})}{s_{\lambda^- \cup \mu^-}(1^{d-1})} \cdot I_{\lambda^-} \otimes I_{\mu'} - S^{(d-1,r-1)}_{\lambda^-\mu'} \bigg).
    \end{align*}
    The second step above follows by applying \Cref{lem:lift_recursion_identity,lem:lift_recursion_integral}, with $\lambda' \leftarrow \lambda^-$. The third step uses \Cref{lem:lambda_minus_identity}. The term appearing inside the summand's lift can be rewritten in terms of $\Delta_{\lambda^- \mu'}$:
    \begin{equation} \label{eq:almost_there}
        \frac{s_{\lambda^-\cup \mu^-}(1^{r-1})}{s_{\lambda^- \cup \mu^-}(1^{d-1})} \cdot I_{\lambda^-} \otimes I_{\mu'} - S^{(d-1,r-1)}_{\lambda^-\mu'}  = \bigg(\frac{s_{\lambda^-\cup \mu^-}(1^{r-1})}{s_{\lambda^- \cup \mu^-}(1^{d-1})} - \frac{s_{\lambda^-\cup \mu'}(1^{r-1})}{s_{\lambda^- \cup \mu'}(1^{d-1})} \bigg) \cdot I_{\lambda^-} \otimes I_{\mu'} + \Delta_{\lambda^- \mu'}.
    \end{equation}
    The term on the far right is in $\Sep_+(\calQ^{d-1}_{\lambda^-} : \calQ^{d-1}_{\mu'})$ by our inductive hypothesis, and therefore lifts to an element of $\Sep_+(\calQ^d_\lambda: \calQ^d_\mu)$. Technically, our inductive hypothesis does not cover those $\mu$ for which $\ell(\mu') > r-1$, but for such diagrams, $\Delta^{(d-1,r-1)}_{\lambda-\mu'} = 0$. The identity does too, provided its prefactor is nonnegative. Thus, it suffices to show that
    \begin{equation} \label{eq:almost_almost_there}
        \frac{s_{\lambda^-\cup \mu^-}(1^{r-1})}{s_{\lambda^- \cup \mu^-}(1^{d-1})} \geq  \frac{s_{\lambda^-\cup \mu'}(1^{r-1})}{s_{\lambda^- \cup \mu'}(1^{d-1})}.
    \end{equation}
    To prove this, we note that if $\mu' \preceq \mu$, then $\mu_{i+1} \leq \mu'_{i} \leq \mu_i$ for all $i \in [d-1]$. In particular, $\mu^-_{i} = \mu_{i+1} \leq \mu'_i$ for all $i \in [d-1]$, and hence $\mu^- \subseteq \mu'$. Thus $\lambda^- \cup \mu^- \subseteq \lambda^- \cup \mu'$ for all $\mu'$. \Cref{lem:ratio_of_s's_subset_diagrams} then implies \Cref{eq:almost_almost_there}, and this completes the proof. \qedhere

\end{proof}

%% file: transposition.tex
\newcommand{\PTT}{\mathrm{PTT}}

Combining the random purification channel with the optimal pure state transposition map of \cite{BGSQ26} straightforwardly gives a simple rank-$r$ mixed state transposer. We write $\calT^{(d,n,k)}: \states(\C^d)^{\otimes n} \to \states(\C^d)^{\otimes k}$ for the optimal pure state transposer.

Two quick remarks before we describe the algorithm. First, the transposition task here is defined with respect to the computational basis. This should not be confused with the transpose in the Schur basis that arises due to the passage to the Choi state later in the lower bound analysis. Second, we define the computational basis of the purification space $\C^d \otimes \C^r$ using a product basis whose first factor is the computational basis of $\C^d$. With this convention, $\Tr_{\reg{B}}( \ketbra{\brho}_{\reg{AB}}^T) = \rho_{\reg{A}}^T$. 

{
\floatstyle{boxed} 
\restylefloat{figure}
\begin{figure}[H]
Given $n$ copies of a rank-$r$ mixed state $\rho_{\reg{A}}$:
\begin{enumerate}
    \item Apply $\Purify^{(d,r,n)}$ to prepare $n$ copies of a random purification $\ket{\brho}_{\reg{AB}} \in \C^d \otimes \C^r$. Set $D = d \cdot r$. 
    \item Apply $\calT^{(D, n, k)}$ yielding a mixed state $\bsigma_{\reg{AB}} \in \states( \C^d \otimes \C^r)^{\otimes k}$.
    \item Trace out the auxiliary registers, and output $\Tr_{\reg{B}}(\bsigma_{\reg{AB}}) \in \states(\C^d)^{\otimes k}$. 
\end{enumerate}
\caption{A mixed state transposition channel, which we call the \emph{PTT transposer}, due to its action: \emph{purify--transpose--trace}. We denote this channel by $\PTT^{(d,n,k)}$. As usual, we will sometimes drop the superscript's parameters when these are clear from context.}
\label{fig:reduction_transpose}
\end{figure}
}

This construction gives us an upper bound on the sample complexity of mixed state transposing.

\begin{proposition}[Sample complexity of the PTT transposer] \label{prop:upper_bound_PTT}
    With $n = O(krd/\eps)$ copies of a mixed state $\rho \in \states(\C^d)$, the PTT transposer $\PTT^{(d,n,k)}$ produces a $k$-copy state with fidelity at least $1-\eps$ with $(\rho^T)^{\otimes k}$. 
\end{proposition}

\begin{proof}
    The proof is analogous to the proof of the sample complexity of the PCT cloner, \Cref{prop:upper_bound_PCT}. We start by writing
    \begin{equation*}
        \PTT(\rho^{\otimes n}) = \Tr_{\reg{B}} \Big(\calT  \big(\Purify \big(\rho^{\otimes n}\big) \big) \Big) = \Tr_{\reg{B}} \Big(\calT  \Big(\E_{\ket{\brho}} \big[\ketbra{\brho}^{\otimes n}\big] \Big) \Big) = \E_{\ket{\brho}} \Big[ \Tr_{\reg{B}} \Big( \calT \big(\ketbra{\brho}^{\otimes n} \big) \Big) \Big].
    \end{equation*}
    By concavity of the ``square-root fidelity'' \cite[Corollary 3.26]{Wat18}, we then have
    \begin{equation*}
        \sqrt{\Fid \Big(  \PTT(\rho^{\otimes n}) , \rho^{T,\otimes k}\Big)} \geq \E_{\ket{\brho}}\Bigg[ \sqrt{\Fid \Big( \Tr_{\reg{B}} \Big( \calT \big(\ketbra{\brho}^{\otimes n} \big) \Big), \rho^{T, \otimes k} \Big)} \Bigg].
    \end{equation*}
    Next, by data processing of fidelity \cite[Theorem 3.27]{Wat18}, 
    \begin{align*}
        \Fid \Big( \Tr_{\reg{B}} \Big( \calT\big(\ketbra{\brho}^{\otimes n} \big) \Big), \rho^{T, \otimes k} \Big) & = \Fid \Big( \Tr_{\reg{B}} \Big( \calT \big(\ketbra{\brho}^{\otimes n} \big) \Big), \Tr_{\reg{B}}\big(\ketbra{\brho}^{T, \otimes k}\big) \Big) \\
        & \geq \Fid \Big( \calT \big(\ketbra{\brho}^{\otimes n} \big) , \ketbra{\brho}^{T, \otimes k}\Big).
    \end{align*}
    In \cite{BGSQ26}, the authors showed that $\calT^{(d,n,k)}$ attains fidelity $d[n]/d[n+k]$ for all pure state inputs. Combining this with \Cref{lem:ratio_of_symmetric_subspace_dimensions}, we obtain
    \begin{equation*}
        \Fid \Big( \calT \big(\ketbra{\brho}^{\otimes n} \big) , \ketbra{\brho}^{T, \otimes k}\Big) = \frac{D[n]}{D[n+k]} \geq 1 - \frac{kD}{n} = 1 - \frac{krd}{n}. 
    \end{equation*}
    Since this holds for all $\ket{\brho}$, we obtain
    \begin{equation*}
        \Fid \Big(  \PTT(\rho^{\otimes n}) , \rho^{T, \otimes k}\Big) \geq \Bigg( \E_{\ket{\brho}} \Bigg[ \sqrt{\Fid \Big( \calT \big(\ketbra{\brho}^{\otimes n} \big) , \ketbra{\brho}^{T, \otimes k}\Big)} \Bigg] \Bigg)^2 \geq 1 - \frac{krd}{n}. 
    \end{equation*}
    Taking $n = \lceil krd/\eps \rceil$ makes this at least $1 - \eps$. This completes the proof.
\end{proof}

We are also able to prove a lower bound in a similar fashion to our cloning lower bound. In this case, the argument turns out to be easier: because the object we want to approximate is transposed, it will turn out that we want an upper bound on 
\begin{equation*}
    \int q_\lambda^d(P) \otimes q_\mu^d(P) \cdot \dP,
\end{equation*}
rather than the partially-transposed integral. As previously noted in the technical overview, this integral is easier to bound in the PSD order directly --- we will see this pan out in the proof below. 

\begin{proposition}[A lower bound for transposing projector states]\label{prop:projector_lower_bound_transpose}
Any channel $\calT: \states(\C^d)^{\otimes n} \to \states(\C^d)^{\otimes k}$ which transposes rank-$r$ projector states to fidelity $1-\eps$ requires at least $n = \Omega(krd/\epsilon)$ copies as input, for $d \geq 2$, $r \leq d/2$, and $\epsilon \leq 1/16$. 
\end{proposition}

\begin{proof}
    The proof follows the same steps as the proof of the cloning lower bound, \Cref{prop:projector_lower_bound}. Fortunately, we have already done almost all of the work necessary in that proof. 

    \paragraph{Step 1:} We start by reducing from fidelity to overlap, using \Cref{lem:reduction_from_fidelity_to_overlap}. This time, we set $\Pi \leftarrow (P^T)^{\otimes k}$ and $\sigma \leftarrow \calT(\rho^{\otimes n})$, finding
    \begin{equation*}
        \Fid \Big( \calT(\rho^{\otimes n}), \rho^{T,\otimes k} \Big) \leq \tr \Big( \calT(\rho^{\otimes n}) \cdot P^{T,\otimes k} \Big). 
    \end{equation*}
    In particular, it suffices to show that if $\tr(\calT (\rho^{\otimes n} ) \cdot P^{T,\otimes k} ) \geq 1-\eps$ for all $\rho$, then we have a lower bound on $n$. 

    \paragraph{Step 2:} By an argument identical to \Cref{lem:reduction_to_perm_invariant}, except with the replacement $P^{\otimes m} \to P^{T,\otimes k}$, we reduce to the case where $\calT$ is a permutation-invariant channel. 

    \paragraph{Step 3:} By an argument identical to \Cref{lem:symmetric_channels_perfomance_upper_bounded}, and the following discussion, except with the replacement $P^{\otimes m} \to P^{T, \otimes k}$, we obtain
    \begin{align}
        \min_P \Big[ \tr\Big( \calT(\rho^{\otimes n}) \cdot P^{T,\otimes k} \Big) \Big] & \leq \sum_{\substack{\lambda \vdash n \\ \ell(\lambda) \leq r}} \frac{\dim(\calP_\lambda) \cdot \dim(\calQ^d_\lambda)}{r^n} \cdot \max_{ \substack{\mu \vdash k \\ \ell(\mu) \leq r}} \norm{ \int_P q_\lambda(P)^T \otimes q_\mu(P)^T \cdot \dP }_\infty \nonumber \\
        & = \E_{\blambda} \bigg[ \frac{s_{\blambda}(1^d)}{s_{\blambda}(1^r)} \cdot \max_{ \substack{\mu \vdash k \\ \ell(\mu) \leq r}} \norm{ \int_P q_{\blambda}(P) \otimes q_\mu(P) \cdot \dP }_\infty \bigg]. \label{eq:min_leq_E_transpose}
    \end{align}

    \paragraph{Step 4:} We now show the following $\mu$-independent bound:
    \begin{equation*}
        \int q^d_\lambda(P) \otimes q^d_\mu(P) \cdot \dP \preceq \frac{s_{\lambda + k\cdot e_1}(1^r)}{s_{\lambda + k \cdot e_1}(1^d)} \cdot I_\lambda \otimes I_\mu.
    \end{equation*}
    There may be a tighter bound which depends on $\mu$, but this bound is easy to show (and has already appeared in \cite{SSW25}), and suffices for our purposes. This follows from the Littlewood--Richardson rule \cite[Corollary 8.3.2(c)]{Ful97}, which tells us that the product representation $q^d_\lambda \otimes q^d_\mu$ decomposes into irreps as
    \begin{equation*}
        q^d_\lambda \otimes q^d_\mu \cong \bigoplus_{\tau} c^\tau_{\lambda \mu} \cdot q^d_\tau.
    \end{equation*}
    Here, $c^\tau_{\lambda \mu} \in \N$ is only possibly nonzero if $\lambda \subseteq \tau$, and $|\tau| = |\lambda| + |\mu| = n+k$. More can be said about these coefficients (see \cite[Section 5.1]{Ful97}), but this is all we will need. We then have
    \begin{equation*}
        \int q^d_\lambda(P) \otimes q^d_\mu(P) \cdot \dP \cong \bigoplus_{\tau} I_{c^{\tau}_{\lambda \mu}} \otimes \bigg(\int q^d_\tau(P) \cdot \dP\bigg) = \bigoplus_\tau \frac{s_\tau(1^r)}{s_\tau(1^d)} \cdot I_{c^{\tau}_{\lambda \mu}} \otimes I_\tau. 
    \end{equation*}
    Here, the ``$\cong$" indicates equality up to a change-of-basis unitary. From this we have
    \begin{equation*}
        \int q^d_\lambda(P) \otimes q^d_\mu(P) \cdot \dP \preceq \bigg(\max_{\tau} \frac{s_\tau(1^r)}{s_\tau(1^d)}\bigg) \cdot I_{\lambda} \otimes I_\mu
    \end{equation*}
    where the maximum is in principle over $\tau$ with $c^\tau_{\lambda \mu}$ nonzero, but can be relaxed to $\tau$ only satisfying $\lambda \subseteq \tau$ and $|\tau \setminus \lambda| = k$. By an argument similar to \Cref{lem:best_mu_is_lambda+ke1}, we can show the maximum is then achieved at $\tau = \lambda+k\cdot e_1$. By the hook-content formula:
    \begin{equation*}
        \frac{s_\lambda(1^d)}{s_\lambda(1^r)} \cdot \frac{s_\tau(1^r)}{s_\tau(1^d)} = \prod_{\Box \in \tau \setminus \lambda} \frac{r+c(\Box)}{d+c(\Box)}.
    \end{equation*}
    Note the ratio is well-defined, since $\ell(\lambda) \leq r$. If we now order the boxes, so that after inserting the first $j$ boxes we have a valid Young diagram, for all $j \in [k]$, then the $j$-th box has content at most $(\lambda_1+j) - 1$, and hence
    \begin{equation*}
        \frac{s_\lambda(1^d)}{s_\lambda(1^r)} \cdot \frac{s_\tau(1^r)}{s_\tau(1^d)} \leq \prod_{j=1}^{k} \frac{r+(\lambda_1+j-1)}{d+(\lambda_1+j-1)}.
    \end{equation*}
    However, this upper bound is attained when $\tau = \lambda+k\cdot e_1$, proving optimality. Thus, 
    \begin{equation}
        \int q^d_\lambda(P) \otimes q^d_\mu(P) \cdot \dP \preceq \frac{s_{\lambda + k\cdot e_1}(1^r)}{s_{\lambda + k \cdot e_1}(1^d)} \cdot I_{\lambda} \otimes I_\mu = \frac{s_{\lambda}(1^r)}{s_\lambda(1^d)} \cdot \bigg(\prod_{j=1}^{k} \frac{r+(\lambda_1+j-1)}{d+(\lambda_1+j-1)} \bigg) \cdot I_\lambda \otimes I_\mu. \label{eq:integral_bound_transpose}
    \end{equation}

    \paragraph{Step 5:} Finally, we combine our ingredients to prove the proposition. From \Cref{eq:min_leq_E_transpose,eq:integral_bound_transpose}, we have
    \begin{equation*}
    \min_{P} \Big[ \tr\Big( \calT(\rho^{\otimes n}) \cdot P^{T,\otimes k} \Big) \Big] \leq \E_{\blambda} \bigg[ \prod_{j=1}^{k} \frac{r + \blambda_1 + j - 1}{d + \blambda_1 + j - 1} \bigg].
    \end{equation*}
    If $\calT$ achieves worst-case overlap $1-\eps$, for all inputs, then the left-hand side is greater than $1-\eps$. Finally, from \Cref{lem:lower_bound_from_WSS}, we obtain a lower bound of 
    \begin{equation*}
        n \geq \frac{1}{8} \cdot \frac{krd}{\eps},
    \end{equation*}
    for $d \geq 2$, $r \leq d/2$, and $\eps \leq 1/16$. This completes the proof. \qedhere
\end{proof}

Lastly, we use our upper and lower bounds to conclude the sample complexity of mixed state transposition.

\begin{theorem}[The PTT transposer is sample-optimal; \Cref{thm:main_result_transposition_intro}, restated.] 
    Suppose $\calT: \states(\C^d)^{\otimes n} \to \states(\C^d)^{\otimes k}$ is a transposition channel such that
    \begin{equation*}
         \Fid \big(  \calT(\rho^{\otimes n}) , \rho^{T,\otimes k}\big) \geq 1-\eps,
    \end{equation*}
    for all states $\rho$ of rank at most $r$. Then we must have $n = \Omega(krd/\eps)$, for $d \geq 2$, and $\epsilon \leq 1/16$. Moreover, the PTT transposer achieves this guarantee with $n = O(krd/\eps)$ copies. Thus, the sample complexity of producing $k$ copies of the transpose of any unknown rank-$r$ input state $\rho^{\otimes n}$ to fidelity $1-\eps$ is $\Theta(krd/\eps)$. 
\end{theorem}

\begin{proof}
    In \Cref{prop:upper_bound_PTT}, we show the PTT transposer attains the upper bound. In \Cref{prop:projector_lower_bound_transpose}, we show that even when promised the input is a rank-$r$ projector state, we require $n = \Omega(krd/\eps)$ copies for $d \geq 2$, $r \leq d/2$, and $\epsilon \leq 1/16$. The lower bound for generic rank follows by considering projector states of rank $\min(r, d/2)$. This completes the proof.
\end{proof}